\documentclass[sigconf,nonacm]{acmart}

\usepackage{booktabs}
\usepackage[utf8]{inputenc}
\usepackage{xspace}
\usepackage{amsmath,amsthm,mathtools,mathrsfs,color,url,bm}
\usepackage{multirow}
\usepackage{footnote}
\usepackage[whole]{bxcjkjatype}
\usepackage{comment}
\usepackage{tikz}
\usetikzlibrary{arrows,shapes,calc}
\usepackage[linesnumbered,ruled,vlined]{algorithm2e}
\usepackage[subrefformat=parens]{subcaption}
\mathtoolsset{showonlyrefs}
\usepackage[referable]{threeparttablex}

\newcommand{\dens}{\texttt{DENSITY}\xspace}
\DeclareMathOperator*{\argmin}{arg\,min}
\DeclareMathOperator*{\argmax}{arg\,max}
\DeclareMathOperator{\supp}{supp}
\newcommand{\ot}{\leftarrow}

\AtBeginDocument{%
  \providecommand\BibTeX{{%
    \normalfont B\kern-0.5em{\scshape i\kern-0.25em b}\kern-0.8em\TeX}}}

\copyrightyear{2023}
\acmYear{2023}
\setcopyright{acmcopyright}\acmConference[WSDM '23]{Proceedings of the Sixteenth ACM International Conference on Web Search and Data Mining}{February 27--March 3, 2023}{Singapore, Singapore}
\acmBooktitle{Proceedings of the Sixteenth ACM International Conference on Web Search and Data Mining (WSDM '23), February 27--March 3, 2023, Singapore, Singapore}
\acmPrice{15.00}
\acmDOI{10.1145/3539597.3570444}
\acmISBN{978-1-4503-9407-9/23/02}

\acmSubmissionID{}

\begin{document}

\title{Stochastic Solutions for Dense Subgraph Discovery\\ in Multilayer Networks}  

%
\author{Yasushi Kawase}
\affiliation{%
  \institution{The University of Tokyo}
  \streetaddress{Hongo 7-3-1}
  \state{Tokyo}
  \country{Japan}
  \postcode{113-8654}
}
\email{kawase@mist.i.u-tokyo.ac.jp}
\author{Atsushi Miyauchi}
\affiliation{%
  \institution{The University of Tokyo}
  \streetaddress{Hongo 7-3-1}
  \state{Tokyo}
  \country{Japan}
  \postcode{113-8654}
}
\email{miyauchi@mist.i.u-tokyo.ac.jp}
\author{Hanna Sumita}
\affiliation{%
  \institution{Tokyo Institute of Technology}
  \streetaddress{Oookayama 2-12-1}
  \state{Tokyo}
  \country{Japan}}
\email{sumita@c.titech.ac.jp}


\begin{abstract}  
  Network analysis has played a key role in knowledge discovery and data mining.
  In many real-world applications in recent years, we are interested in mining \emph{multilayer networks},
  where we have a number of edge sets called \emph{layers},
  which encode different types of connections and/or time-dependent connections over the same set of vertices.
  Among many network analysis techniques, dense subgraph discovery, aiming to find a dense component in a network, is an essential primitive with a variety of applications in diverse domains.
  In this paper, we introduce a novel optimization model for dense subgraph discovery in multilayer networks.
  Our model aims to find a stochastic solution, i.e., a probability distribution over the family of vertex subsets,
  rather than a single vertex subset, whereas it can also be used for obtaining a single vertex subset.
  For our model, we design an LP-based polynomial-time exact algorithm.
  Moreover, to handle large-scale networks, we also devise a simple, scalable preprocessing algorithm, which often reduces the size of the input networks significantly and results in a substantial speed-up.
  Computational experiments demonstrate the validity of our model and the effectiveness of our algorithms.
\end{abstract}

\ccsdesc[500]{Theory of computation~Graph algorithms analysis}
\ccsdesc[500]{Theory of computation~Mathematical optimization}

\keywords{network analysis, multilayer networks, dense subgraph discovery, stochastic solutions}  %

\maketitle

\section{Introduction}\label{sec:introduction}
Network analysis has played a key role in knowledge discovery and data mining. 
In many real-world applications in recent years, we are interested in mining \emph{multilayer networks} rather than ordinary (i.e., single-layer) networks, where we have a number of edge sets called \emph{layers}, which encode different types of connections and/or time-dependent connections over the same set of vertices~\cite{Boccaletti+14,DeDomenico+13,Kivela+16}.
For example, in the Twitter network, there are various layers representing different types of connections, e.g., follower--followee relations, retweets, and mentions, among users.
Moreover, each of those connections is time-dependent and therefore leads to multiple layers by itself.
As another example, consider brain networks arising in neuroscience, where vertices correspond to small regions of a brain~\cite{Friston11}.
In this network, we can obtain at least two layers representing the structural connectivity and the functional connectivity (e.g., co-activation) among the small pieces of a brain.

Among many network analysis techniques, dense subgraph discovery, aiming to find a dense component in a network, is an essential primitive with a variety of applications in diverse domains~\cite{Gionis_Tsourakakis_15,Lee+10}.
Examples include
detecting communities and spam link farms in Web graphs~\cite{Dourisboure+_07,Gibson+_05},
experts extraction in crowdsourcing systems~\cite{Kawase+19},
real-time story identification in microblogging streams~\cite{Angel+_12},
and extracting molecular complexes in protein--protein interaction networks~\cite{Bader_Hogue_03}.

\begin{figure}[t]
  \centering
  \begin{minipage}{0.67\linewidth}
    \includegraphics[scale=.18]{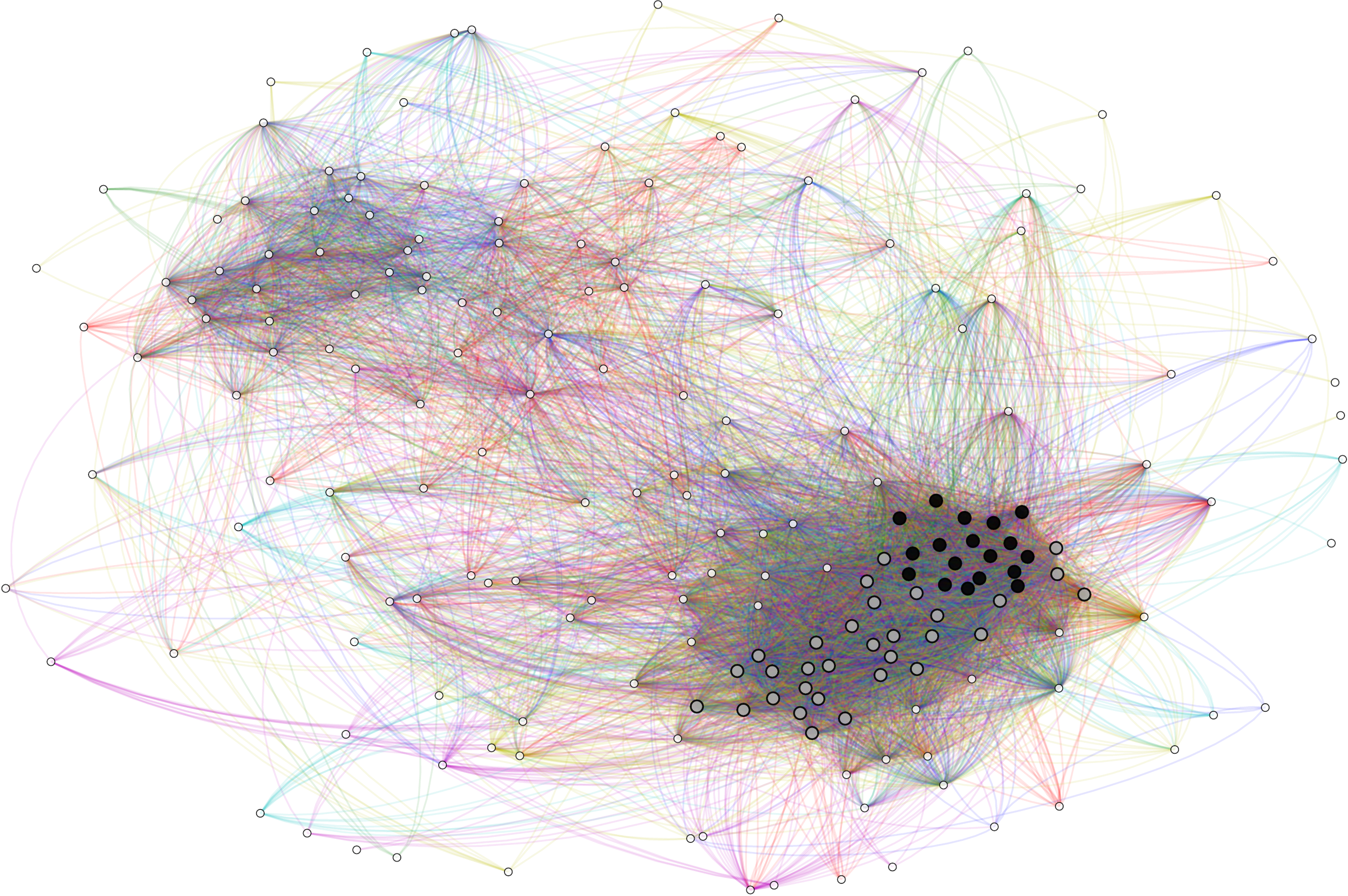}
  \end{minipage}
  \begin{minipage}{0.32\linewidth}
    \includegraphics[scale=.044]{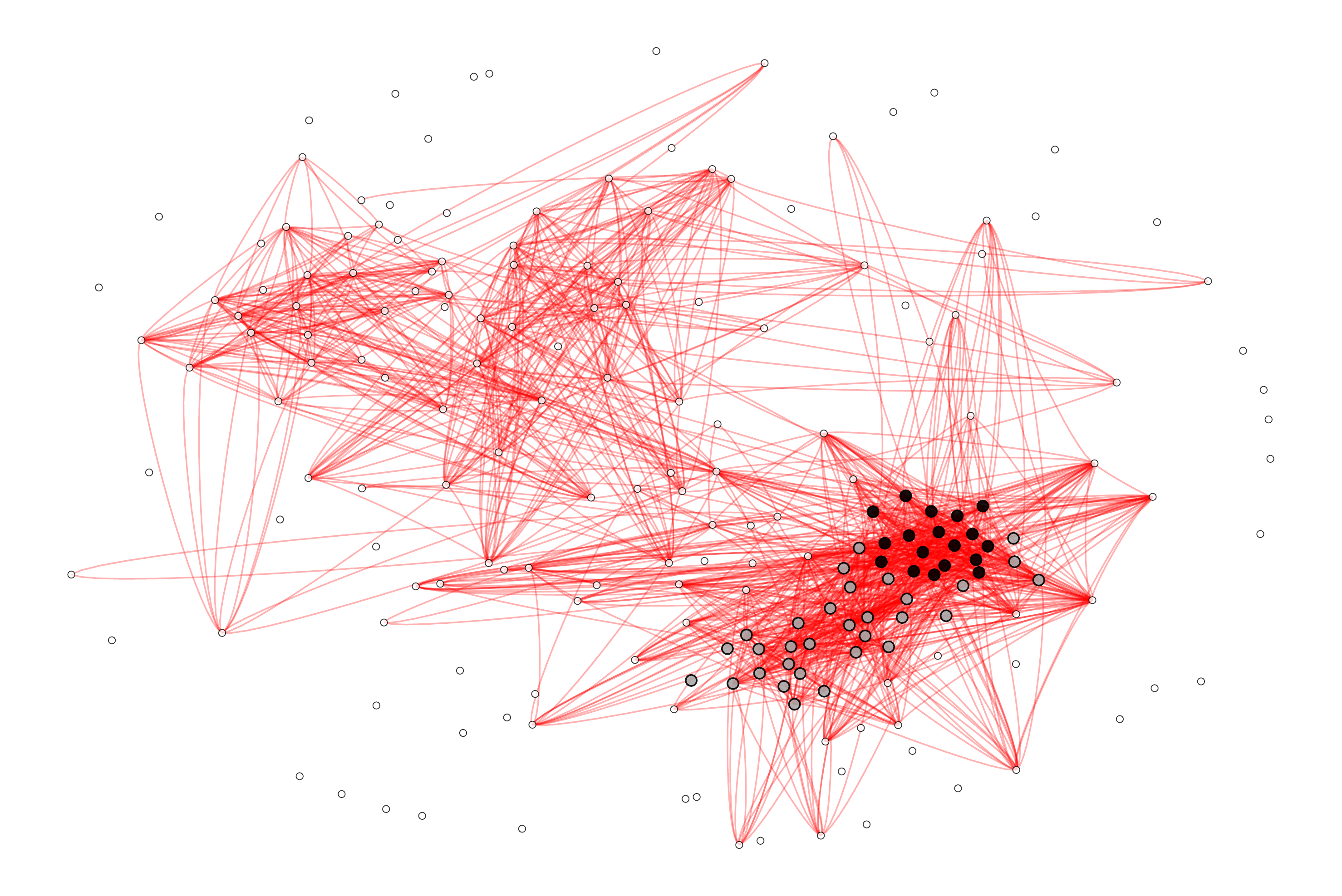}
    \includegraphics[scale=.044]{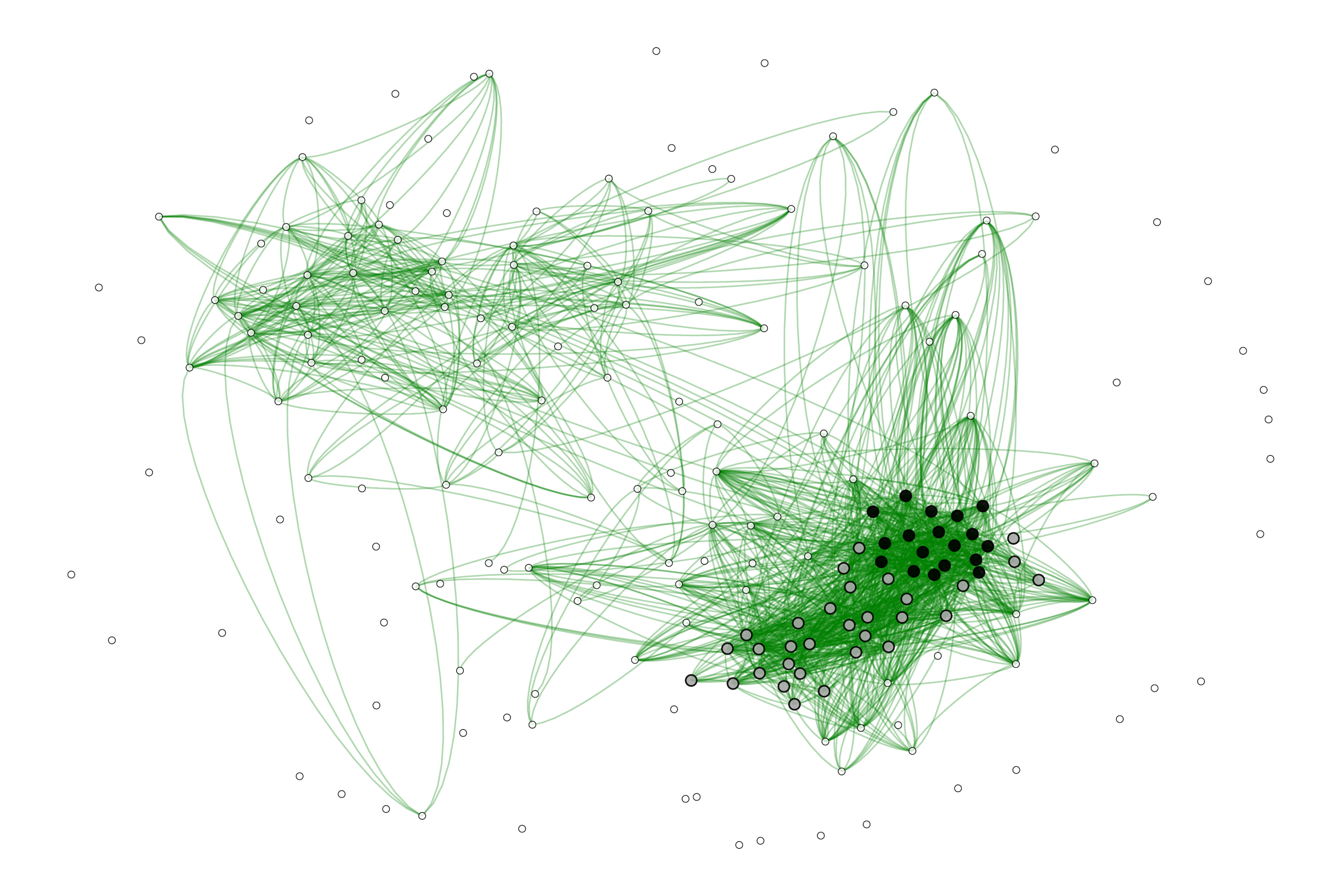}\\
    \includegraphics[scale=.044]{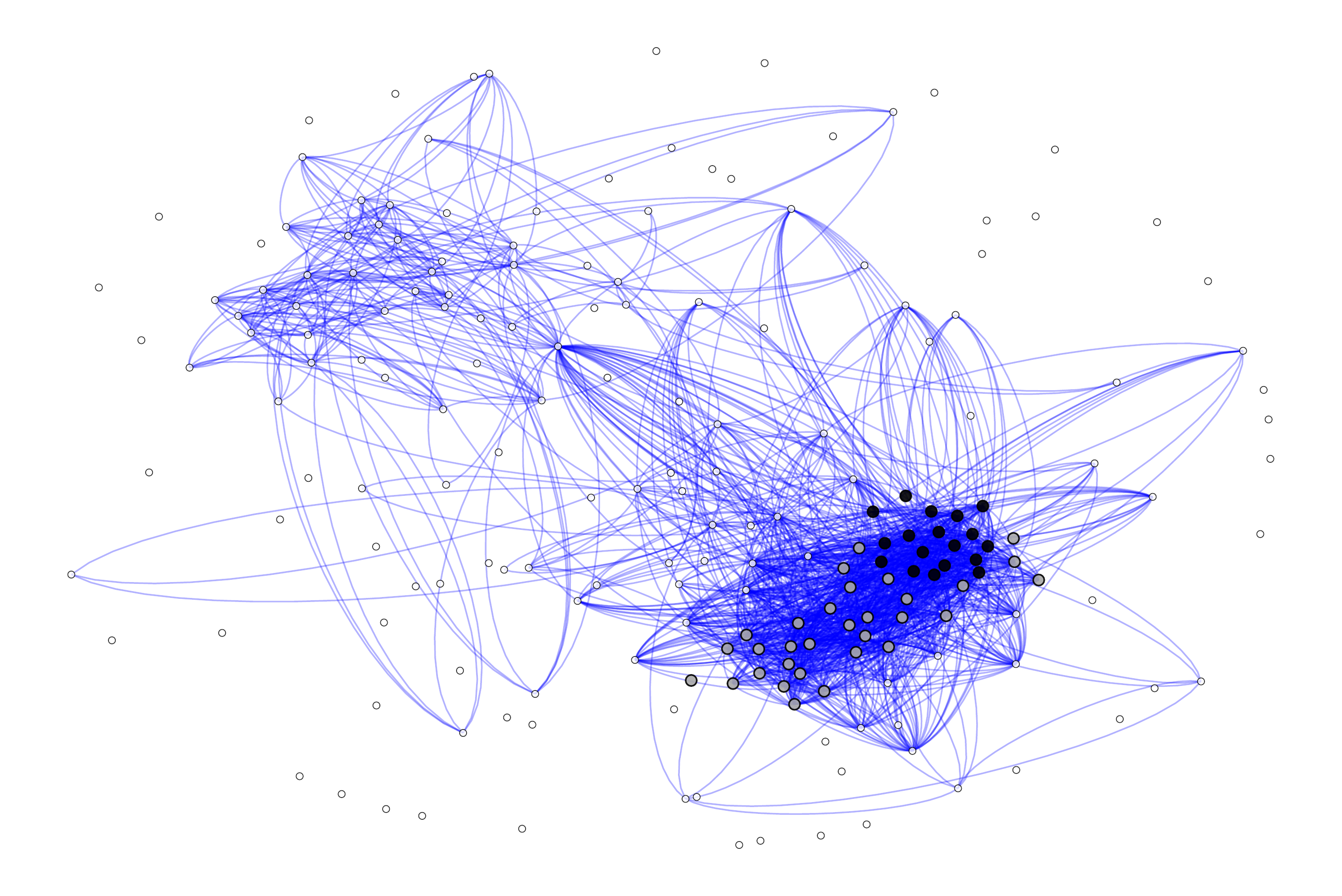}
    \includegraphics[scale=.044]{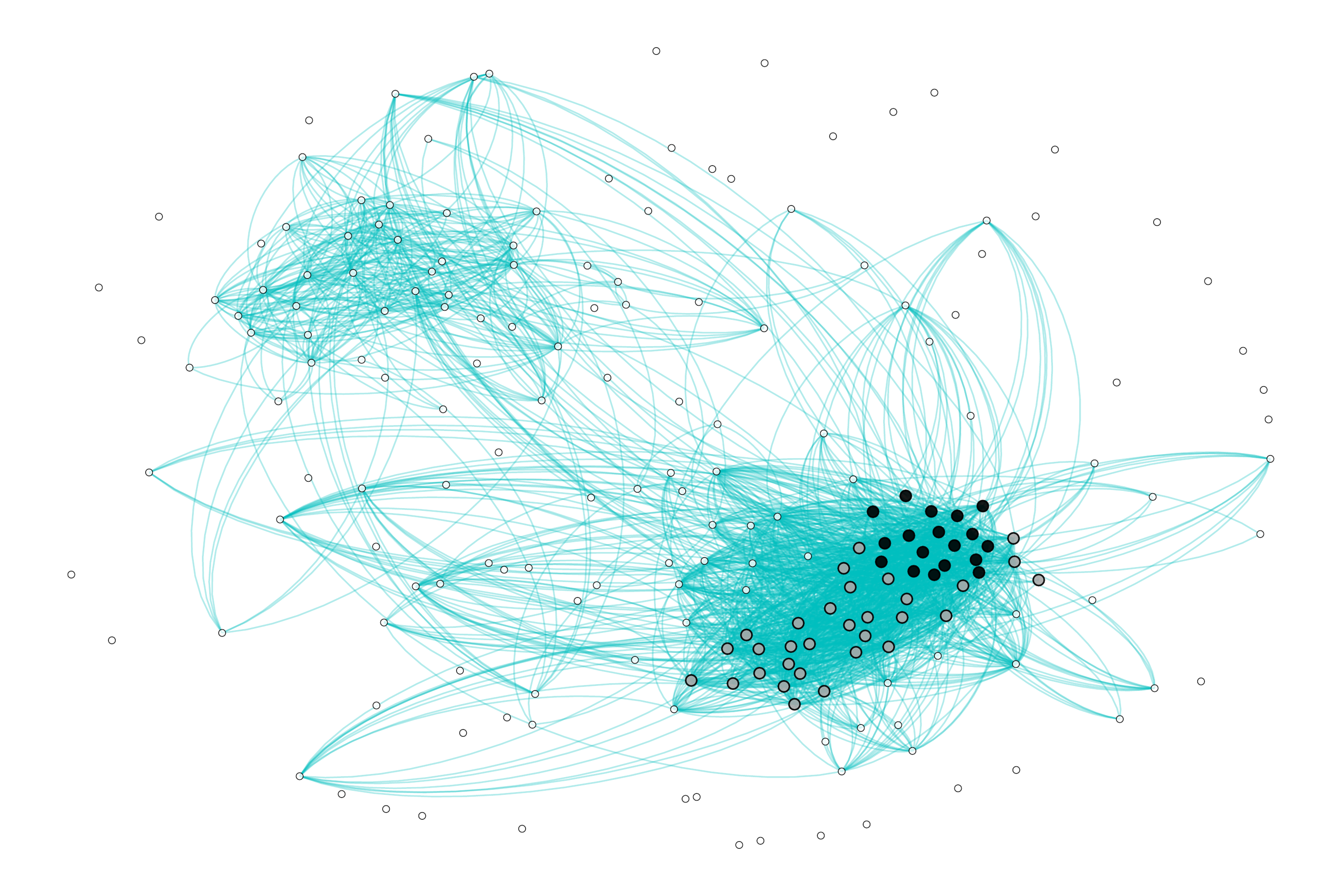}\\
    \includegraphics[scale=.044]{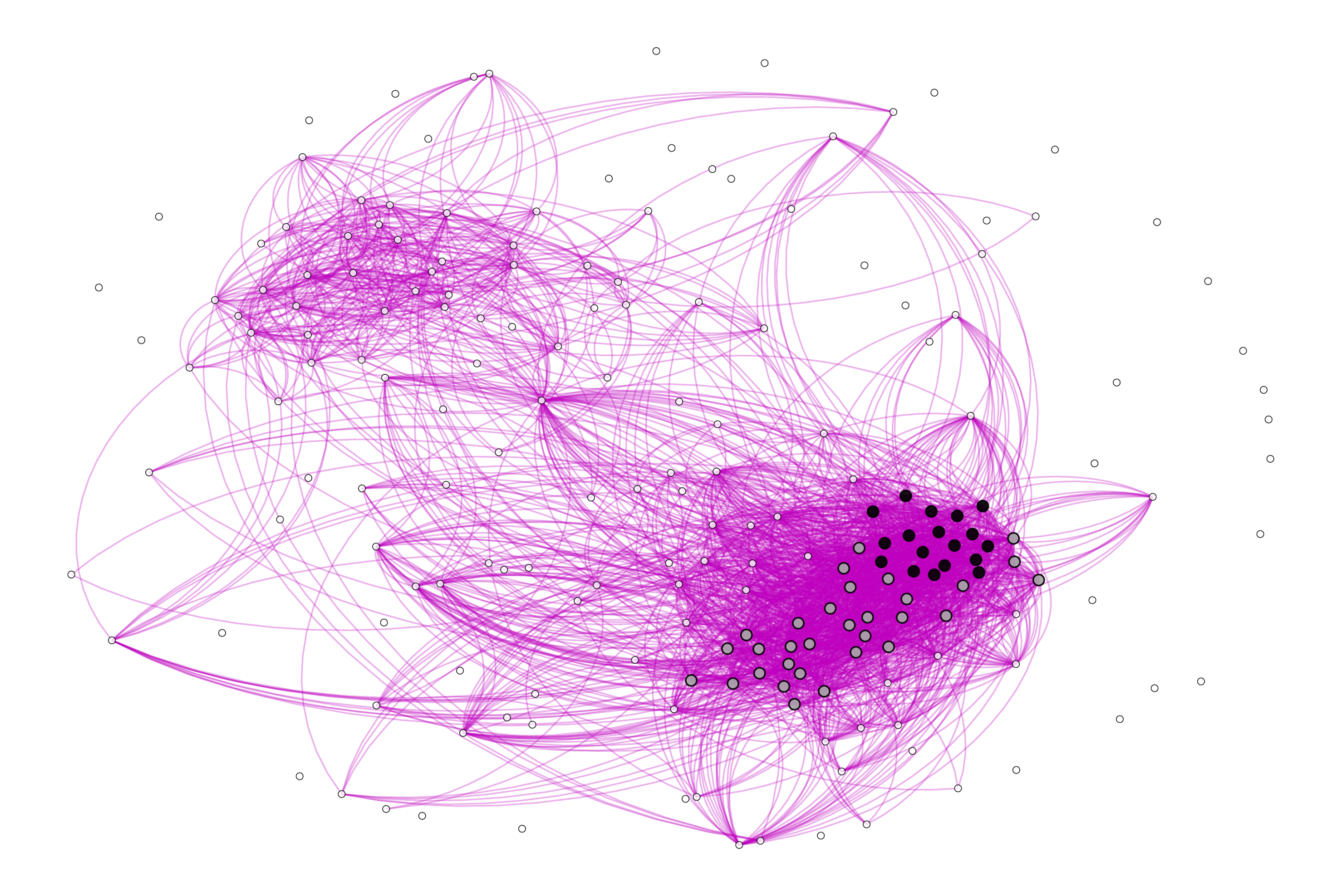}
    \includegraphics[scale=.044]{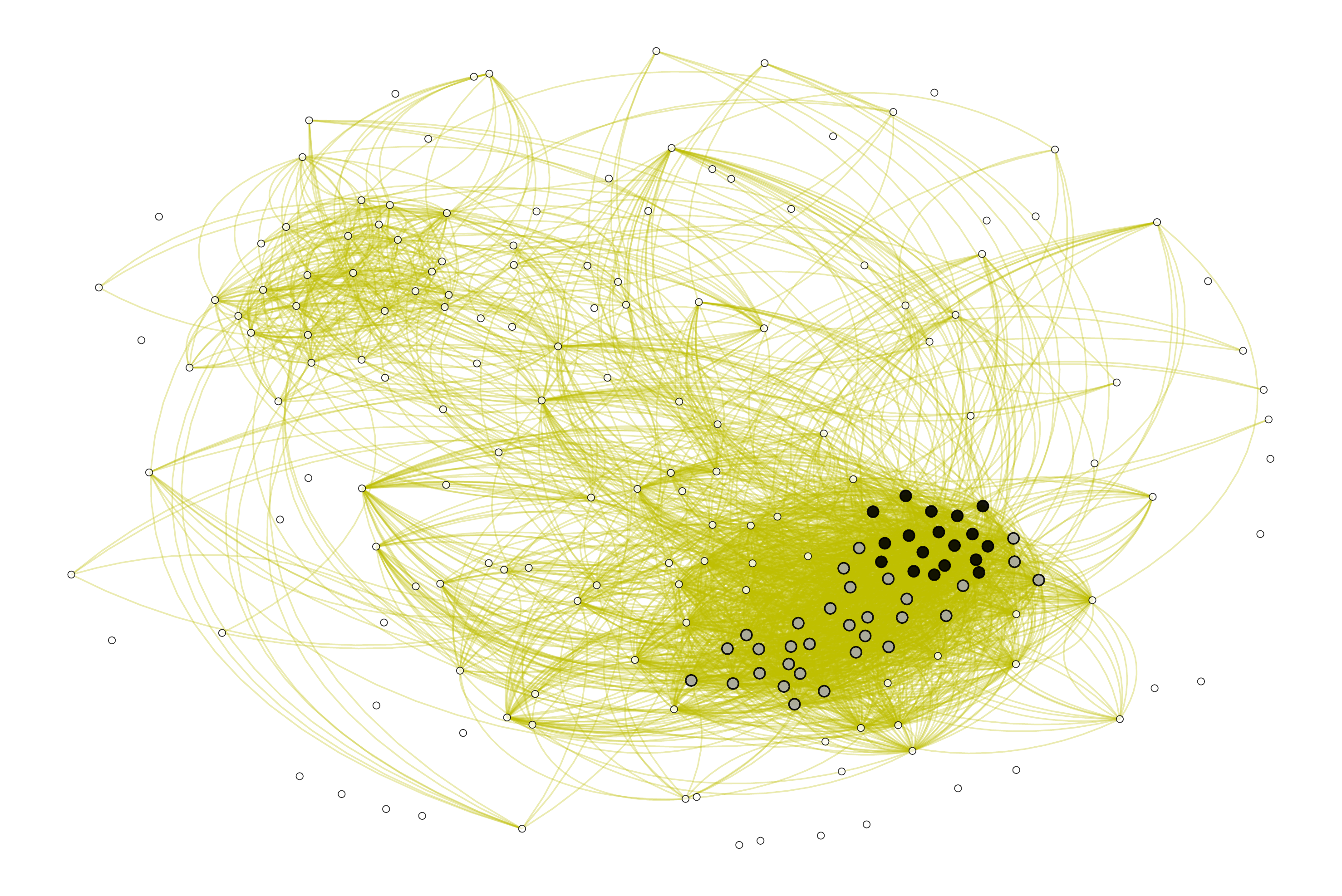}
  \end{minipage}
  \caption{An optimal solution to our model with the regret metric for the real-world multilayer network called WILDBIRDS: selecting the black vertices with probability 0.29 and the union of the gray and black vertices with probability 0.71. There are six layers, and the color of each edge represents the layer containing the edge.}
  \label{fig:result_example}
\end{figure}

Recently, dense subgraph discovery has been extended from single-layer networks to multilayer networks.
Jethava and Beerenwinkel~\cite{JB2015} introduced the \emph{densest common subgraph problem},
where given a multilayer network, we are asked to find a vertex subset
that maximizes the minimum \emph{degree density} (defined later) over the layers.
They mentioned some concrete applications of this problem in biological networks.
Moreover, we can see that the densest common subgraph problem may appear in the context of robust optimization.
Indeed, multilayer networks can be seen as a model of single-layer networks with uncertain edges with a number of scenarios,
where each layer corresponds to one scenario.
If we can find a subgraph that is reasonably dense for all layers,
the solution is more robust than that obtained by single-layer network analysis,
which may be arbitrarily bad for some layers.
%
%

\subsection{Our contribution}
In this paper, we introduce a novel optimization model for dense subgraph discovery in multilayer networks.
Given an edge-weighted multilayer network, our model aims to find a vertex subset that is dense for the layer selected by an adversary.
Specifically, we employ a stochastic solution, i.e., a probability distribution over the family of vertex subsets, while existing work focuses only on a deterministic solution, i.e., a single vertex subset.
More precisely, we stochastically choose a vertex subset according to a stochastic solution and then,
knowing only the stochastic solution but not a realization, the adversary selects the worst layer for us.
Note that the worst layer with respect to each realization of the stochastic solution may be different from that with respect to the stochastic solution.
It is worth mentioning that our model can be seen as a zero-sum two-player Stackelberg game~\cite{BO1999}, where the leader is our algorithm and the follower is the adversary.

We measure the density of a vertex subset using the quality function called the \emph{degree density}
(or simply \emph{density}) similarly to Jethava and Beerenwinkel~\cite{JB2015}.
The degree density of a vertex subset (in a single-layer network) is defined as half the average degree of the subgraph induced by the subset.
This quality function is often used in the literature of dense subgraph discovery; in fact, this is the objective function of the well-known \emph{densest subgraph problem} (see Section~\ref{sec:related}).

We evaluate the performance of a stochastic solution (or an algorithm)
using the following three metrics for the layer selected by the adversary:
(i) the \emph{density}, i.e., the expected degree density,
(ii) the \emph{robust ratio}, i.e., the ratio of the expected density to the optimal density, and
(iii) the \emph{regret}, i.e., the difference between the optimal density and the expected density.
Therefore, we address three optimization problems according to the metrics.
Our main algorithmic contribution is to design a polynomial-time exact algorithm for our optimization model,
which is applicable to all of the above three metrics.
Specifically, the algorithm first solves an LP, which is a generalization of the LP used for the densest common subgraph problem~\cite{JB2015}
and then computes an optimal probability distribution based on the LP's optimal solution.
%

We observe that the output of our algorithm has a useful structure;
the family of the vertex subsets with positive probabilities has a hierarchical structure.
This leads to several practical benefits, e.g., the largest size subset contains all the other subsets
and the optimal solution obtained by our algorithm has support size at most the number of vertices (although a solution of our optimization model may have support size exponential in the number of vertices).
Moreover, we can demonstrate that the support size of the solution obtained is upper bounded by the number of layers.
It is desirable to have a small support size
for the understandability of stochastic solutions and for purposes of simple verification and validation.

For practical use of our proposed algorithm,
we wish to speed up our algorithm
so that it is applicable to larger-scale multilayer networks.
The bottleneck is the computation cost of the LP, which has a lot of variables and constraints.
To this end, we devise a simple, scalable preprocessing algorithm, which often reduces the size of the input networks significantly and results in a substantial speed-up.
Specifically, the algorithm first computes an approximate solution by solving a much smaller LP than the aforementioned LP
and then removes vertices from the original network using the information of the approximate solution obtained.
Note that our preprocessing algorithm is a generalization of that proposed by Balalau et al.~\cite{Balalau+15} for the densest subgraph problem.
However, to verify that the preprocessing algorithm does not harm any optimal solution, we need a more sophisticated analysis, which is totally different from theirs.

Finally, we conduct thorough computational experiments using synthetic graphs and real-world networks to verify the validity of our model and to evaluate the performance of our algorithms.
Figure~\ref{fig:result_example} illustrates an optimal solution to our model with the regret metric for the real-world multilayer network called WILDBIRDS (see Section~\ref{sec:experiments} for its detailed characteristics).
The optimal solution obtained is reasonably dense for all layers, as desired.

It should be remarked that stochastic solutions can also be used for obtaining a single vertex subset.
For example, the following three rules are reasonable to select a vertex subset from a stochastic solution:
(i) stochastically select a subset following the probability distribution,
(ii) select a subset with the highest probability,
and
(iii) select a subset with non-zero probability that optimizes some metric at hand (e.g., the minimum density, minimum robust ratio, and maximum regret, over layers).
In our experiments, we compare the vertex subsets detected using the above rules with those obtained by existing algorithms for the densest common subgraph problem,
in terms of various evaluation metrics.

\section{Related Work}\label{sec:related}
The densest subgraph problem is one of the most popular optimization models for dense subgraph discovery.
Let $G=(V,E)$ be an undirected graph and $w\colon E\to\mathbb{R}_{++}$ a positive edge weight.
For a vertex subset $S\subseteq V$, the subgraph induced by $S$ is denoted by $G[S]\coloneqq(S,E[S])$, where $E[S]\coloneqq \bigl\{\{u,v\}\in E\mid u,v\in S\bigr\}$.
In addition, we denote by $w(S)$ the total weight of the edges in $S$, i.e., $w(S)=\sum_{e\in E[S]}w(e)$.
For a nonempty $S\subseteq V$, the \emph{degree density} (or simply called \emph{density}) of $S$ is defined as $w(S)/|S|$ (where we define the density of the empty set (i.e., $0/0$) to be $0$).
In the densest subgraph problem, given a graph $G=(V,E)$ with an edge weight $w$,
we are asked to find $S\subseteq V$ that maximizes the density $w(S)/|S|$.

It is well known that the densest subgraph problem can be solved exactly in polynomial time using a maximum-flow-based algorithm~\cite{Goldberg_84} or an LP-based algorithm~\cite{Charikar2000}.
Moreover, it was shown that a simple greedy algorithm called the \emph{greedy peeling} admits $2$-approximation in $O(m+n\log n)$ time~\cite{Charikar2000,Kortsarz+94}. 
Recently, Boob et al.~\cite{Boob+20} designed an iterative greedy peeling algorithm, and demonstrated empirically that the output tends to be nearly optimal. 
Later, Chekuri, Quanrud, and Torres~\cite{Chekuri+22} proved the convergence to optimality of this algorithm (in a more general context). 

Balalau et al.~\cite{Balalau+15} introduced a simple preprocessing algorithm to improve the scalability of (exact) algorithms for the densest subgraph problem.
Their preprocessing algorithm first computes an approximate solution $S\subseteq V$ (using the greedy peeling),
and then iteratively removes a vertex with the (weighted) degree less than the objective value of the approximate solution obtained.
The validity of this preprocessing algorithm is guaranteed by the fact that any vertex with the (weighted) degree less than the optimal value is not contained in any optimal solution~\cite{Balalau+15}.

For the densest common subgraph problem, Jethava and Beerenwinkel~\cite{JB2015} devised an LP-based polynomial-time heuristic and a $2k$-approximation algorithm based on the greedy peeling, where $k$ is the number of layers.
Later, Charikar, Naamad, and Wu~\cite{charikar+18} designed two polynomial-time algorithms with approximation ratios $O(\sqrt{|V|\log k})$ and $O(|V|^{2/3})$ (irrespective of $k$), respectively.
Moreover, they showed some strong inapproximability results for the problem, based on some reasonable computational complexity assumptions.
Thus, it is very unlikely that a well-approximate solution can be found in polynomial time.
In contrast to this, as mentioned above, we can compute an \emph{optimal} stochastic solution in terms of the aforementioned three metrics in polynomial time.

Recently, Galimberti, Bonchi, and Gullo~\cite{Galimberti+_17} introduced a generalization of the densest common subgraph problem,
which they refer to as the \emph{multilayer densest subgraph problem}.
This problem exploits a trade-off between the minimum density value over layers and the number of layers considered.
They proposed an approximation algorithm using a core decomposition technique for multilayer networks.
Very recently, Hashemi, Behrouz, and Lakshmanan~\cite{Hashemi+22} designed a sophisticated core decomposition algorithm,
which they call the \emph{FirmCore decomposition algorithm}.
Their algorithm finds the set of $(k,\lambda)$-FirmCores for all possible $k$ and $\lambda$ in polynomial time,
where $(k,\lambda)$-FirmCore is a maximal subgraph in which every vertex has degree no less than $k$ in the subgraph for at least $\lambda$ layers.
They demonstrated that the decomposition unfolds a better solution to the multilayer densest subgraph problem than the algorithm by Galimberti, Bonchi, and Gullo~\cite{Galimberti+_17} for many instances.

Semertzidis et al.~\cite{Semertzidis+19} introduced another generalization of the densest common subgraph problem, called the Best Friends Forever (BFF) problem,
in the context of evolving graphs with a number of snapshots.
The BFF problem is a series of optimization problems that maximize an \emph{aggregate density} over snapshots,
where the aggregate density is set to be the average/minimum value of the average/minimum degree of vertices over layers.
They investigated the computational complexity of the problems and designed some approximation or heuristic algorithms.

Recalling that multilayer networks can also be seen as a model of networks with \emph{uncertainty},
we can find some other optimization models related.
Zou~\cite{Zou_13} studied the densest subgraph problem in \emph{uncertain graphs}.
An uncertain graph is a pair of $G=(V,E)$ and $p\colon E\rightarrow [0,1]$, 
where $e\in E$ is present with probability $p(e)$ whereas $e\in E$ is absent with probability $1-p(e)$.
In the problem, given an uncertain graph $G=(V,E)$ with $p$, we seek $S\subseteq V$ that maximizes the expected density.
Zou~\cite{Zou_13} showed that this problem can be reduced to the original densest subgraph problem,
and designed a polynomial-time exact algorithm based on the reduction.
Recently, Tsourakakis et al.~\cite{Tsourakakis+19} introduced a more general optimization problem called the \emph{risk-averse dense subgraph discovery}.

As another example, Miyauchi and Takeda~\cite{Miyauchi_Takeda_18} introduced an optimization problem called the \emph{robust densest subgraph problem}.
In this problem, given an undirected graph $G=(V,E)$ and an edge-weight space $I=\times_{e\in E}[l_e,r_e]\subseteq \times_{e\in E}[0,\infty)$, we are asked to find $S\subseteq V$ that maximizes $\min_{w\in I} \frac{w(S)/|S|}{w(S^*_w)/|S^*_w|}$, where $S^*_w$ is an optimal solution to the densest subgraph problem for $G$ with $w$.
The intuition of this problem is the same as that of ours; this problem also seeks $S\subseteq V$ that is reasonably dense for any $w\in I$.
However, they considered only deterministic solutions and gave a strong hardness result.

As well as dense subgraph discovery, many important primitives for single-layer network analysis have recently been extended to multilayer networks.
Examples include community detection~\cite{Bazzi+16,DeBacco+17,Interdonato+17,Tagarelli+17}, link prediction~\cite{DeBacco+17,Jalili+17}, analyzing spreading processes~\cite{DeDomenico+16,Salehi+15}, and identifying central vertices~\cite{Basaras+19,DeDomenico+15}. 


\section{Model}\label{sec:model}
In this section, we formally define our optimization model. 
Let $G=(V, (E_i)_{i\in [k]})$ be a multilayer network consisting of $k$ layers,
and let $w_i\colon E_i\to\mathbb{R}_{++}$ be a positive edge weight for layer $i$.
We denote by $E$ the union of all edge sets, i.e., $E\coloneqq \bigcup_{i\in [k]} E_i$.
Let $S_i^*$ be a densest subgraph for layer $i$, i.e., $S_i^*\in\argmax_{S\subseteq V}w_i(S)/|S|$.
Our task is to find a vertex subset that is dense for the layer selected adversarially.
As mentioned in the introduction, we consider a stochastic solution to compete with the adversary.
Let $\Delta(2^V)$ be the set of probability distributions over $2^V$.
For each $p \in \Delta(2^V)$, we denote by $p_S$ the probability of choosing $S \subseteq V$.

We aim to compute $p\in \Delta(2^V)$ that maximizes some metric (when the adversary selects the worst layer to $p$).
We employ the following three metrics:

\smallskip
\noindent\textbf{Density.}
\ The first metric is the degree density itself.
Specifically, when we select a vertex subset according to a probability distribution $p \in \Delta(2^V)$ and the adversary selects a layer $i\in[k]$,
our metric is defined as follows:
\begin{align}
   & 
  \mathbb{E}_{S\sim p}\left[ \frac{w_i(S)}{|S|}\right]
  \quad\Bigg(=
  \sum_{S\subseteq V} p_S \frac{w_i(S)}{|S|}\Bigg).
  \label{eq:mindens}
\end{align}
As the adversary selects the worst layer,
we aim to find $p\in \Delta(2^V)$ that maximizes the minimum of the density~\eqref{eq:mindens} among $i\in [k]$.
Our optimization model with this metric can be seen as a stochastic version of the densest common subgraph problem introduced by Jethava and Beerenwinkel~\cite{JB2015}.


\smallskip
\noindent\textbf{Robust ratio.}
\ The \emph{robust ratio} is a metric based on the ratio of the expected degree to the optimal degree. 
For a nonempty $S\subseteq V$ and $i\in[k]$, let us consider a normalized density defined as $\frac{w_i(S)/|S|}{w_i(S_i^*)/|S_i^*|}$.
When we select a vertex subset according to a probability distribution $p\in \Delta(2^V)$ and the adversary selects a layer $i\in [k]$,
the metric is defined as follows:
\begin{align}
  \mathbb{E}_{S\sim p}\left[ \frac{w_i(S)/|S|}{w_i(S_i^*)/|S_i^*|}\right]
  \quad\Bigg(=
  \sum_{S\subseteq V}p_S \frac{w_i(S)/|S|}{w_i(S_i^*)|S_i^*|}\Bigg).
  \label{eq:ratio}
\end{align}
In other words, the robust ratio is equivalent to
the density for the multilayer network with weights $w_1',\dots,w_k'$ given by $w_i'(e)\coloneqq \frac{w_i(e)}{w_i(S_i^*)/|S_i^*|}$ for each $i\in[k]$ and $e \in E_i$.
As the adversary selects the worst layer,
we aim to find $p\in \Delta(2^V)$ that maximizes the minimum of~\eqref{eq:ratio} among $i\in [k]$.
Note that the optimal robust ratio is contained in the interval $[1/k,\,1]$ 
because $p\in \Delta(2^V)$ such that $p_{S_i^*}=1/k$ for each $i\in [k]$ has the objective value of $1/k$.

\smallskip
\noindent\textbf{Regret.}
\ The \emph{regret} is a metric based on the difference between the optimal density and the expected density.
For $S\subseteq V$ and $i\in [k]$, the regret is defined as $w_i(S^*_i)/|S^*_i|-w_i(S)/|S|$.
When we select a vertex subset according to a probability distribution $p\in \Delta(2^V)$ and the adversary selects a layer $i\in [k]$,
the metric is defined as follows:
\begin{align}
  \mathbb{E}_{S\sim p}\left[ \frac{w_i(S_i^*)}{|S_i^*|}-\frac{w_i(S)}{|S|}\right]
  \quad\Bigg(=
  \frac{w_i(S_i^*)}{|S_i^*|}-\sum_{S\subseteq V}p_S \frac{w_i(S)}{|S|}\Bigg).
  \label{eq:regret}
\end{align}
As the adversary selects the worst layer,
we aim to find $p\in \Delta(2^V)$ that minimizes the maximum of~\eqref{eq:regret} among $i\in [k]$.

\smallskip
Here we explain how to select an appropriate metric.
The density and regret metrics are useful when we are concerned with multilayer networks with homogeneous layers 
such as time-dependent follower-followee relations in the Twitter network. 
Although the density metric can be the first choice, the regret metric is more suitable for robust analysis. 
For example, consider the case where there are a number of layers consistent with each other together with some noisy (e.g., random) layers. 
The density metric would suffer from the effect of the noisy layers, but the regret metric would avoid it and find dense subgraphs in the other meaningful layers. 
On the other hand, the robust ratio metric is useful when we analyze multilayer networks with heterogeneous layers 
such as brain networks with structural and functional connectivity layers.
From its definition, the robust ratio metric would find subgraphs that are reasonably dense for all layers. 
The density and regret metrics focus only on the layers with small optimal densities and the layers with large optimal densities, respectively.

\smallskip
\noindent\textbf{Unified concept: $(\bm{\alpha},\bm{\beta})$-density.}
\ Here we introduce a general metric, enabling us to deal with the above three metrics in a unified manner.
An important fact is that the robust ratio and regret metrics can be obtained by affine transformations of the density.
Specifically, when we select a subgraph according to a probability distribution $p \in \Delta(2^V)$ and the adversary selects a layer $i\in[k]$,
we define the \emph{$(\bm{\alpha},\bm{\beta})$-density} using two vectors $\bm{\alpha} \in \mathbb{R}^k_+$ and $\bm{\beta} \in \mathbb{R}^k$ as follows:
\begin{align}
  \mathbb{E}_{S\sim p}\left[\alpha_i\frac{w_i(S)}{|S|}+\beta_i\right]
  \quad\Bigg(=
  \alpha_i\sum_{S\subseteq V}p_S\frac{w_i(S)}{|S|}+\beta_i\Bigg).
  \label{eq:general}
\end{align}
Note that the above three metrics, the density, robust ratio, and regret, are equivalent to
the $(\bm{1},\bm{0})$-density,
$((|S_i^*|/w_i(S_i^*))_{i\in[k]},\bm{0})$-density, and
$(\bm{1},(-w_i(S_i^*)/|S_i^*|)_{i\in[k]})$-density\footnote{The actual regret value is the negation of $(\bm{1},(-w_i(S_i^*)/|S_i^*|)_{i\in[k]})$-density.}, respectively. 
Note that $w_i(S_i^*)/|S_i^*|$ is polynomially computable for each $i\in [k]$~\cite{Charikar2000,Goldberg_84}.
We refer to $(\bm{\alpha},\bm{\beta})$-\dens as the problem of finding $p\in \Delta(2^V)$ that maximizes the minimum of the $(\bm{\alpha},\bm{\beta})$-density among $i\in [k]$.
Therefore, in the following, we aim to design an algorithm for $(\bm{\alpha},\bm{\beta})$-\dens.

\section{Algorithm}\label{sec:algorithm}
In this section, we provide an LP-based polynomial-time exact algorithm for $(\bm{\alpha},\bm{\beta})$-\dens.
Let $((x_e)_{e\in E},(y_v)_{v\in V},t)$ be continuous variables.
We consider the following LP:
\begin{align}
  \begin{array}{rll}
    \text{max.} & t                                                                      &                                   \\[3pt]
    \text{s.t.} & \displaystyle  t\leq \alpha_i\cdot \sum_{e\in E_i}w_{i}(e) x_e+\beta_i & (\forall i\in [k]),               \\[3pt]
                & x_e\le y_u,\ x_e\le y_v                                                & (\forall e=\{u,v\}\in E),         \\[3pt]
                & \displaystyle \sum_{v\in V}y_{v}= 1,                                   &                                   \\[3pt]
                & x_e,y_v\ge 0                                                           & (\forall e\in E, \forall v\in V).
  \end{array}\label{LP:general}
\end{align}
When $(\bm{\alpha},\bm{\beta})=(\bm{1},\bm{0})$,
this formulation coincides with the LP introduced by Jethava and Beerenwinkel~\cite{JB2015} for the densest common subgraph problem.
However, they did not point out the connection between a solution of the LP and a distribution in $\Delta(2^V)$.

Our algorithm first computes an optimal solution to LP~\eqref{LP:general},
denoted by $((\hat{x}_e)_{e\in E},(\hat{y}_v)_{v\in V},\hat{t})$.
Then we set $r_0,r_1,\dots,r_\ell$ to be the reals such that $\{r_0,r_1,\dots,r_\ell\}=\{\hat{y}_v\mid v\in V\}\cup\{0\}$ and $r_0~(=0)<r_1<r_2<\cdots<r_\ell$.
In addition, let $S_j=\{v\in V\mid \hat{y}_v\ge r_j\}$ $(j=1,\dots,\ell)$.
The output of our algorithm is the probability distribution $\hat{p}\in\Delta(2^V)$ defined as
\begin{align}\label{eq:def of p}
&\hat{p}_{S_j} = (r_j-r_{j-1})\cdot |S_j| \ (j \in [\ell]) \  \text{ and } \nonumber  \\
&\hat{p}_S = 0  \text{ for the other $S$'s}.
\end{align}
We remark that $\sum_{S\in 2^V}\hat{p}_S = 1$ holds because
\[1=\sum_{v \in V} \hat{y}_v = \sum_{j\in [\ell]} \sum_{v:\,\hat{y}_v=r_j}r_j = \sum_{j\in [\ell]} r_j(|S_j|-|S_{j+1}|) = \sum_{j\in [\ell]} \hat{p}_{S_j},\]
where $S_{\ell+1}=\emptyset$ for ease of notation.

Our algorithm is formally described in Algorithm~\ref{alg:general}.
Clearly, the algorithm runs in polynomial time.

\begin{algorithm}[t]
  \caption{LP-based algorithm}\label{alg:general}
  Solve LP~\eqref{LP:general} to obtain an optimal solution $((\hat{x}_e)_{e\in E},(\hat{y}_v)_{v\in V},\hat{t})$\;
  Let $r_0,r_1,\dots,r_\ell$ be reals such that $\{r_0,r_1,\dots,r_\ell\}=\{\hat{y}_v\mid v\in V\}\cup\{0\}$ and $r_0~(=0)<r_1<r_2<\cdots<r_\ell$\;
  Let $S_j= \{v\in V\mid \hat{y}_v\ge r_j\}$ $(j=1,\dots,\ell)$\;
  \Return $\hat{p}\in\Delta(2^V)$ with $\hat{p}_{S_j}=(r_j-r_{j-1})\cdot |S_j|$ ($j=1,\dots,\ell$) (and $\hat{p}_{S}=0$ for the other $S$'s).
\end{algorithm}

\section{Analysis}

In this section, we first demonstrate that the output $\hat{p}$ of Algorithm~\ref{alg:general} is optimal to $(\bm{\alpha},\bm{\beta})$-\dens.
Then we analyze the structure of the output $\hat{p}$ with a special attention to its support size.

\subsection{Optimality of the output of Algorithm~\ref{alg:general}}
We first prove that the $(\bm{\alpha},\bm{\beta})$-density 
of $\hat{p}$ is equal to the optimal value of LP~\eqref{LP:general}:
\begin{lemma}\label{lemma:expect}
  It holds that
  \begin{align*}
    \min_{i\in [k]}\mathbb{E}_{S\sim \hat{p}}\Bigl[ \alpha_i\frac{w_i(S)}{|S|}+\beta_i\Bigr]
    = \min_{i\in [k]}\Bigl[\alpha_i\sum_{e\in E_i}w_i(e)\hat{x}_e+\beta_i\Bigr].
  \end{align*}
\end{lemma}
\begin{proof}
  To prove the lemma, we show that for any $i\in [k]$,
  \begin{align*}
    \mathbb{E}_{S\sim \hat{p}}\Bigl[ \alpha_i\frac{w_i(S)}{|S|}+\beta_i\Bigr] = \alpha_i\sum_{e\in E_i}w_i(e)\hat{x}_e+\beta_i.
  \end{align*}
  By the definition of $S_j$ $(j=1, \ldots, \ell)$, for each $e =\{u,v\} \in E_i$,  we observe that $e\in E_i[S_j] \iff u, v \in S_j \iff \hat{y}_u\ge r_j\text{ and }\hat{y}_v\ge r_j\iff \hat{x}_e\ge r_j$,
  where the last equivalence follows from $\hat{x}_e = \min\{ \hat{y}_u, \hat{y}_v\}$.

  Recall that $\hat{p}$ is defined as \eqref{eq:def of p}.
  Then we have
  \begin{align*}
    \mathbb{E}_{S\sim \hat{p}}\Bigl[ \alpha_i\frac{w_i(S)}{|S|}+\beta_i\Bigr]
     & =\alpha_i\sum_{j\in[\ell]}\frac{w_i(S_j)}{|S_j|}\cdot \hat{p}_{S_j}+\beta_i                    \\
     & =\alpha_i\sum_{j\in[\ell]} w_i(S_j)\cdot (r_j-r_{j-1})+\beta_i                                 \\
     & =\alpha_i\sum_{j\in[\ell]} \sum_{e\in E_i[S_j]}w_i(e)\cdot (r_j-r_{j-1})+\beta_i               \\
     & =\alpha_i\sum_{j\in[\ell]} \sum_{e\in E_i:\,\hat{x}_e\ge r_j}w_i(e)\cdot (r_j-r_{j-1})+\beta_i \\
     & =\alpha_i\sum_{e\in E_i}w_i(e)\hat{x}_e+\beta_i,
  \end{align*}
  where the second last equality follows from the above equivalence and the last equality follows from the fact that for each $e\in E$, $\hat{x}_e = r_{j'} = \sum_{j=1}^{j'} (r_j - r_{j-1})$ for some $j'$.
\end{proof}

Next, we prove that the optimal value of LP~\eqref{LP:general} gives an upper bound on the optimal value of $(\bm{\alpha},\bm{\beta})$-\dens:
\begin{lemma}\label{lemma:general_lower}
  It holds that
  \begin{align*}
    \min_{i\in [k]}\Bigl[\alpha_i\sum_{e\in E_i}w_i(e)\hat{x}_e+\beta_i\Bigr]\geq
    \max_{p\in\Delta(2^V)}\min_{i\in [k]}\mathbb{E}_{S\sim p}\Bigl[ \alpha_i\frac{w_i(S)}{|S|}+\beta_i\Bigr].
  \end{align*}
\end{lemma}

\begin{proof}
  Let us take an arbitrary $\tilde{p} \in \Delta(2^V)$.
  We consider a solution $((\tilde{x}_e)_{e\in E},(\tilde{y}_v)_{v\in V},\tilde{t})$ of LP~\eqref{LP:general} such that
  \begin{align}\label{eq:transform}
    \tilde{x}_e=\!\!\!\!\sum_{\substack{S\subseteq V: \\e\in E[S]}}\frac{\tilde{p}_S}{|S|},\
    \tilde{y}_v=\!\!\!\!\sum_{\substack{S\subseteq V: \\v\in S}}\frac{\tilde{p}_S}{|S|},\
    \tilde{t}=\!\min_{i\in[k]} \Bigl[\alpha_i\!\sum_{e\in E_i}w_i(e)\tilde{x}_e+\beta_i\Bigr].
  \end{align}
  The solution $((\tilde{x}_e)_{e\in E},(\tilde{y}_v)_{v\in V},\tilde{t})$ is feasible for LP~\eqref{LP:general};
  in fact, the only concern is the third constraint but we see that
  \begin{align}
    \sum_{v\in V}\tilde{y}_v
    =\sum_{v\in V}\sum_{\substack{S\subseteq V: \\v\in S}}\frac{\tilde{p}_S}{|S|}
    =\sum_{S\subseteq V}\left(|S|\cdot \frac{\tilde{p}_S}{|S|}\right)
    =1.
  \end{align}
  Hence, the optimal value of LP~\eqref{LP:general} is at least $\tilde{t}$.
  Moreover, we have
  \begin{align*}
    \tilde{t}
     & =\min_{i\in[k]} \Bigl[\alpha_i\sum_{e\in E_i}w_i(e)\tilde{x}_e+\beta_i\Bigr]                       \\
     & =\min_{i\in[k]} \Bigl[\alpha_i \sum_{S\subseteq V}\tilde{p}_S\cdot\frac{w_i(S)}{|S|}+\beta_i\Bigr]
    =\min_{i\in[k]}\mathbb{E}_{S\sim \tilde{p}}\Bigl[\alpha_i\frac{w_i(S)}{|S|}+\beta_i\Bigr].
  \end{align*}
  Recalling that $\tilde{p}$ is taken arbitrarily from $\Delta(2^V)$,
  we see that the optimal value of LP~\eqref{LP:general} is at least that of $(\bm{\alpha},\bm{\beta})$-\dens.
  \qedhere
\end{proof}

Combining Lemmas~\ref{lemma:expect} and \ref{lemma:general_lower}, we have the desired result:
\begin{theorem}
  Algorithm~\ref{alg:general} outputs an optimal solution to $(\bm{\alpha}, \bm{\beta})$-\dens.
\end{theorem}

\subsection{Hierarchical structure and support size of the output of Algorithm~\ref{alg:general}}
Here we observe some useful properties of the output of Algorithm~\ref{alg:general}.
By the design of Algorithm~\ref{alg:general},
we see that the support of the output $\hat{p}$ of the algorithm
(i.e., an optimal solution to $(\bm{\alpha},\bm{\beta})$-\dens) has a hierarchical structure.
We denote the support of $p\in \Delta(2^V)$ by $\supp(p) = \{ S\subseteq V \mid p_S > 0 \}$.
\begin{proposition}\label{prop:chain}
  Algorithm~\ref{alg:general} outputs a hierarchical solution $\hat{p}$,
  i.e., $S\subseteq T$ or $S\supseteq T$ for any $S,T\in\supp(\hat{p})$.
\end{proposition}
From this proposition, the output $\hat{p}$ has support size at most $|V|-1$ (since the empty set and the singletons are useless).
In addition, if we pick an optimal basic solution to LP~\eqref{LP:general} in Algorithm~\ref{alg:general},
the support size of $\hat{p}$ becomes at most $k$. 
For the definition of a basic solution, see e.g., Vanderbei's book~\cite{vanderbei2020linear}.
\begin{theorem}\label{thm:support}
  If Algorithm~\ref{alg:general} takes a basic optimal solution to LP~\eqref{LP:general},
  its output $\hat{p}$ has support size at most $k$.
\end{theorem}

\begin{proof}
  Let $((\hat{x}_e)_{e\in E},(\hat{y}_v)_{v\in V},\hat{t})$ be a basic solution.
  Without loss of generality, we may assume that
  \begin{align}\label{eq:opt sol}
    \hat{x}_e=\min\{\hat{y}_u,\hat{y}_v\} \quad (\forall e=\{u, v\} \in E).
  \end{align}
  Recall that $\ell$ denotes the number of different positive values in $\{\hat{y}_v \mid \hat{y}_v>0, \ v\in V\}$.
  Let $V_0 =\{v\in V\mid\hat{y}_v=0\}$.
  We divide $V \setminus V_0$ into $\ell$ subsets of vertices sharing the same value of $\hat{y}_v$, denoted by $V_1, \ldots, V_\ell$.

  Let us focus on the constraints in LP~\eqref{LP:general} that are satisfied with equality.
  For each $j =0,1, \ldots, \ell$, let $F_j \subseteq E[V_j]$ be a spanning forest in $E[V_j]$.
  Let $\rho$ be the number of connected components in $E[V_1] \cup \cdots \cup E[V_\ell]$,
  and let $\zeta$ be that of $E[V_0]$.
  We arbitrarily take $\zeta$ vertices, denoted by $u_1, \ldots, u_\zeta$, one from each connected component in $E[V_0]$.
  We focus on the following constraints (satisfied with equality):
  \begin{align*}
     & \textstyle t=\alpha_i \sum_{e\in E_i}w_{i}(e) x_e+\beta_i  \quad(\forall i\in K'),                                          \\
     & \textstyle x_e= y_u,\ x_e= y_v  \quad(\forall e=\{u,v\}\in \bigcup_{j=0}^\ell E[V_j]),                                      \\
     & \textstyle x_e= y_v             \quad(\forall e=\{u,v\} \in E \setminus \bigcup_{j=0}^\ell E[V_j],\ \hat{y}_v < \hat{y}_u), \\
     & \textstyle \sum_{v\in V}y_{v}= 1,                                                                                           \\
     & \textstyle y_{v}=0\quad (\forall v \in V_0),
  \end{align*}
  where $K' = \{ i \in [k] \mid t= \alpha_i\cdot \sum_{e\in E_i}w_{i}(e) x_e+\beta_i\}$.

  We prove that the coefficient matrix of those constraints is not full-rank.
  For ease of discussion, we remove constraints that are represented by a linear combination of others.
  Specifically, it is enough to focus on the following constraints:
  \begin{align}
     & \textstyle t=\alpha_i\cdot \sum_{e\in E_i}w_{i}(e) x_e+\beta_i\quad  (\forall i\in K'), \label{eq:constraint1}                      \\
     & \textstyle x_e= y_u,\ x_e= y_v \quad(\forall e=\{u,v\}\in \bigcup_{j=0}^\ell F_j), \label{eq:constraint2}                           \\
     & \textstyle x_e= y_v \quad(\forall e=\{u,v\}\in \bigcup_{j=0}^\ell (E[V_j] \setminus F_j),\ u<v),\label{eq:constraint3}              \\
     & \textstyle x_e= y_v \quad(\forall e=\{u,v\}\in E\setminus \bigcup_{j=0}^\ell E[V_j],\ \hat{y}_v < \hat{y}_u),\label{eq:constraint4} \\
     & \textstyle \sum_{v\in V}y_{v}= 1,\quad\label{eq:constraint5}                                                                        \\
     & \textstyle y_{u_i}=0 \quad(\forall i\in[\zeta]).\label{eq:constraint6}
  \end{align}

  There are two types of missing constraints.
  First, let $e=\{u,v\} \in E[V_j] \setminus F_j$ $(j \in [\ell])$ such that $x_e=y_v$ appears in \eqref{eq:constraint3}.
  There exists a cycle $C$ in $F_j \cup \{e\}$.
  By a telescoping sum of the constraints \eqref{eq:constraint2} along $C$, i.e.,
  $y_v = x_{e'}$, $x_{e'}=y_{v'}$, \ldots, $x_{e''}=y_u$,
  we obtain $y_u=y_v$.
  Then, constraint $x_e= y_u$ is obtained by summing $x_e= y_v$ for $y_u=y_v$.
  Next, let $v\in V_0$ belong to a connected component containing $u_j$ $(j \in [\zeta])$.
  By summing \eqref{eq:constraint2} along a path from $v$ to $u_j$, i.e.,
  $y_v = x_{e'}$, $x_{e'}=y_{v'}$, \ldots, $x_{e''}=y_{u_j}$, we obtain $y_v = y_{u_j}$.
  Then, constraint $y_v=0$ is obtained by summing $y_v=y_{u_j}$ and $y_{u_j}=0$ from \eqref{eq:constraint6}.

  The rank of the coefficient matrix is equal to the number of constraints \eqref{eq:constraint1}--\eqref{eq:constraint6}.
  For \eqref{eq:constraint1}, we have at most $k$ constraints;
  the number of constraints \eqref{eq:constraint2} is $2(|V|-\rho-\zeta)$;
  that for \eqref{eq:constraint3} and \eqref{eq:constraint4} is $|E|-|V|+\rho+\zeta$;
  that for \eqref{eq:constraint5} and \eqref{eq:constraint6} is $1+\zeta$.
  Therefore, we have at most $|V|+|E|-\rho+1+k$ constraints.

  We have $|V|+|E|+1$ variables in LP~\eqref{LP:general}.
  If $\rho > k$, then we have at most $|V|+|E|$ constraints, and hence the coefficient matrix of \eqref{eq:constraint1}--\eqref{eq:constraint6} cannot have rank $|V|+|E|+1$.
  As the solution is basic, $\rho \leq k$ must hold.
  Recall that each $V_j$ $(j=1,\ldots, \ell)$ has at least one connected component.
  Therefore, we have at most $k$ different positive values of $y_v$'s, which implies that the output $\hat{p}$ has support size at most $k$.
\end{proof}

\section{Preprocessing}
In this section, we present a simple, scalable preprocessing algorithm, which often reduces the size of the input networks significantly and results in a substantial speed-up. 
Specifically, the algorithm first computes an approximate solution by solving an LP, 
which is much smaller than LP~\eqref{LP:general} in practice, 
and then removes vertices from the original network using the information of the approximate solution obtained. 
We assume that $S^*_i$ for all $i\in [k]$ are known in advance
because they can be computed efficiently  using Charikar's LP-based algorithm~\cite{Charikar2000}
together with the preprocessing algorithm introduced by Balalau et al.~\cite{Balalau+15}.
To describe our algorithm, we introduce some notations. 
For $S\subseteq V$, $v\in S$, and $i\in[k]$, let $d_i(S,v)$ denote the weighted degree of $v$ in the subgraph induced by $S$ in layer $i$, 
i.e., $d_i(S,v)\coloneqq\sum_{e\in E_i[S]:\,v\in e}w_i(e)$. When $S=V$, we simply write $d_i(v)$.

We first describe a fast algorithm for finding an approximate solution for $(\bm{\alpha}, \bm{\beta})$-\dens. 
Specifically, we compute a probability distribution $q\in\Delta(2^V)$ that maximizes the $(\bm{\alpha}, \bm{\beta})$-density 
under the constraint that $q_S=0$ for all $S\in 2^V\setminus \{S_1^*,\dots,S_k^*\}$.
The distribution can be found by solving the following LP:
\begin{align}
\begin{array}{rll}
\text{max.}       & t &\\
\text{s.t.}& \displaystyle t\leq \alpha_i \sum_{j\in[k]} \frac{w_{i}(S_j^*)}{|S_j^*|} q_j+\beta_i  &(\forall i\in [k]),\\
           & \displaystyle \sum_{j\in[k]}q_j = 1,  &\\
           & q_j\ge 0                &(\forall j\in[k]).
\end{array}\label{LP:lower}
\end{align}
Note that this LP has $k+1$ variables and $2k+1$ constraints. 
As $k$ is usually much smaller than $|V|$ and $|E|$, this LP is much smaller than LP~\eqref{LP:general} in practice. 

Next we describe an algorithm for removing vertices using the information of the above approximate solution $q$. 
Let $\ell^*$ be the $(\bm{\alpha},\bm{\beta})$-density of $q$, i.e., the optimal value of LP~\eqref{LP:lower}. 
Note that this is a lower bound on the optimal value of $(\bm{\alpha},\bm{\beta})$-\dens. 
Our algorithm iteratively removes any vertex $v^*$ that satisfies $\max_{i\in[k]}[\alpha_i\cdot d_i(V',v^*)+\beta_i]< \ell^*$, where $V'$ is a remaining vertex set (initially $V'=V$), as long as there exists such a vertex. 
For reference, we describe the procedure in Algorithm~\ref{alg:remove}. 
This algorithm can be implemented to run in $O(k|E|+|V|\log |V|)$ time.
\begin{algorithm}[t] 
\caption{Remove useless vertices}\label{alg:remove}
\SetKwInOut{Input}{Input}
\SetKwInOut{Output}{Output}
\Input{\ $(V,(E_i)_{i\in [k]})$ with $w_1,\dots,w_k$, and $\ell^*\in \mathbb{R}$}
\Output{\ $(V',(E_i[V'])_{i\in [k]})$} 
$V'\ot V$\;
\While{\texttt{True}}{
  Let $v^*\in\argmin_{v\in V'}\max_{i\in[k]}[\alpha_i\cdot d_i(V',v)+\beta_i]$\;
  \If{$\max_{i\in[k]}[\alpha_i\cdot d_i(V',v^*)+\beta_i]\ge \ell^*$}{
      \Return $(V',(E_i[V'])_{i\in [k]})$. 
  }
  \lElse{$V'\ot V'\setminus\{v^*\}$}
}
\end{algorithm}

From now on, we demonstrate that the algorithm does not remove any vertex that is contained in a subset in $\supp(p)$, where $p$ is an arbitrary 
optimal solution to $(\bm{\alpha},\bm{\beta})$-\dens. 
The following is a key lemma in our analysis. 
\begin{lemma}\label{lemma:6.2}
Let $\ell^*$ be a lower bound on the optimal value of $(\bm{\alpha},\bm{\beta})$-\dens.
If $\max_{i\in[k]}[\alpha_i\cdot d_i(v^*)+\beta_i]<\ell^*$, 
then $y_{v^*}=0$ for any optimal solution to LP~\eqref{LP:general}.
\end{lemma}
\begin{proof}
We prove the lemma by contradiction.
We denote by $((\hat{x}_e)_{e\in E},(\hat{y}_v)_{v\in V},\hat{t})$ an optimal solution to LP~\eqref{LP:general}
and let $v^*\in V$ be a vertex that satisfies $\alpha_i\cdot d_i(v^*)+\beta_i<\ell^*$ for all $i\in[k]$.
Suppose for contradiction that $\hat{y}_{v^*}>0$.

We construct a solution $((x_e)_{e\in E},(y_v)_{v\in V},t)$ of LP~\eqref{LP:general} as follows: 
\begin{align*}
x_e&=\begin{cases}
\frac{1}{1-\hat{y}_{v^*}}\cdot\hat{x}_e&(e\not\ni v^*),\\
0        &(e\ni v^*),
\end{cases}\quad
y_v=\begin{cases}
\frac{1}{1-\hat{y}_{v^*}}\cdot\hat{y}_v&(v\ne v^*),\\
0        &(v= v^*),
\end{cases}\\
t&=\min_{i\in[k]}\Bigl[\alpha_i\sum_{e\in E_i}w_i(e)x_e+\beta_i\Bigr].
\end{align*}
It is easy to see that $((x_e)_{e\in E},(y_v)_{v\in V},t)$ is a feasible solution of LP~\eqref{LP:general}.
Moreover, we have
\begin{align*}
t
&=\min_{i\in[k]}\Bigl[\alpha_i\sum_{e\in E_i}w_i(e)x_e+\beta_i\Bigr]\\
&=\min_{i\in[k]}\Bigl[\alpha_i\frac{\sum_{e\in E_i:\,v^*\not\in e}w_i(e)\hat{x}_e}{1-\hat{y}_{v^*}}+\beta_i\Bigr]\\
&=\min_{i\in[k]}\Bigl[\alpha_i\frac{\sum_{e\in E_i}w_i(e)\hat{x}_e-\sum_{e\in E_i:\, v^*\in e}w_i(e)\hat{x}_{e}}{1-\hat{y}_{v^*}}+\beta_i\Bigr]\\
&\ge \min_{i\in[k]}\Bigl[\alpha_i\frac{\sum_{e\in E_i}w_i(e)\hat{x}_e-\sum_{e\in E_i:\, v^*\in e}w_i(e) \hat{y}_{v^*}}{1-\hat{y}_{v^*}}+\beta_i\Bigr]\\
&=\min_{i\in[k]}\Bigl[\alpha_i\frac{\sum_{e\in E_i}w_i(e)\hat{x}_e-d_i(v^*) \hat{y}_{v^*}}{1-\hat{y}_{v^*}}+\beta_i\Bigr]\\
&=\min_{i\in[k]}\frac{(\alpha_i\sum_{e\in E_i}w_i(e)\hat{x}_e+\beta_i)- (\alpha_i\cdot d_i(v^*)+\beta_i)\cdot \hat{y}_{v^*}}{1-\hat{y}_{v^*}}\\
&>\min_{i\in[k]}\frac{(\alpha_i\sum_{e\in E_i}w_i(e)\hat{x}_e+\beta_i)-\ell^*\cdot \hat{y}_{v^*}}{1-\hat{y}_{v^*}}\\
&=\frac{\hat{t}-\ell^*\cdot \hat{y}_{v^*}}{1-\hat{y}_{v^*}}
\ge \frac{\hat{t}-\hat{t}\cdot \hat{y}_{v^*}}{1-\hat{y}_{v^*}}
=\hat{t},
\end{align*}
where the first inequality follows from $\hat{x}_{e} \leq \hat{y}_{v^*}$ for each $e \ni v^*$, 
the second inequality follows from the assumptions $\alpha_i\cdot d_i(v^*)+\beta_i< \ell^*$ and $\hat{y}_{v^*}>0$, and the third inequality follows from Lemma~\ref{lemma:general_lower}. 
This contradicts the optimality of $((\hat{x}_e)_{e\in E},(\hat{y}_v)_{v\in V},\hat{t})$. 
\end{proof}

\begin{theorem}
Let $\ell^*$ be a lower bound on the optimal value of $(\bm{\alpha},\bm{\beta})$-\dens.
Then, any vertex $v^*$ that satisfies $\max_{i\in[k]}[\alpha_i\cdot d_i(v^*)+\beta_i]<\ell^*$ 
is not contained in any subset in the support of any optimal solution to $(\bm{\alpha},\bm{\beta})$-\dens.
\end{theorem}
\begin{proof}
Let $v^*$ be any vertex with $\max_{i\in[k]}[\alpha_i\cdot d_i(v^*)+\beta_i]<\ell^*$. 
Let $\hat{p}$ be any optimal solution to $(\bm{\alpha},\bm{\beta})$-\dens.
Construct $((\hat{x}_e)_{e\in E},(\hat{y}_v)_{v\in V},\hat{t})$ from $\hat{p}$ as in \eqref{eq:transform}. 
From Lemma~\ref{lemma:expect} and the proof of Lemma~\ref{lemma:general_lower}, 
we see that $((\hat{x}_e)_{e\in E},(\hat{y}_v)_{v\in V},\hat{t})$ is an optimal solution to LP~\eqref{LP:general}. 
Thus, by Lemma~\ref{lemma:6.2}, we have $\hat{y}_{v^*} = 0$.
By the construction of $\hat{y}_{v^*}$ in \eqref{eq:transform}, $\hat{p}_S=0$ for all $S\subseteq V$ containing $v^*$. 
\end{proof}

This theorem indicates that Algorithm~\ref{alg:remove} does not remove any vertex that is contained in a subset in the support of any optimal solution to $(\bm{\alpha},\bm{\beta})$-\dens.

\section{Experimental Evaluation}\label{sec:experiments}
In this section, we conduct computational experiments using synthetic graphs and real-world networks
to verify the validity of our proposed model and
to evaluate the performance of our proposed algorithms.
All experiments were conducted on a machine equipped with Intel Xeon W 10-core processor and 64GB RAM.
Algorithms were implemented in Python using Gurobi Optimizer 9.0.2.

\subsection{Validity of our model}
Here we aim to verify the validity of our model using synthetic graphs.
To this end, we use a randomly generated multilayer network with a planted clique,
and examine whether an optimal solution detects vertex subsets close to the clique.

We first explain our random procedure for generating multilayer networks.
We produce an (unweighted) random power-law graph as a layer using the Chung--Lu model~\cite{CL2002}, where we first specify an expected degree $d_v$ for each $v\in V$ according to the power-law distribution with exponent $\beta$,
and then connect each pair of vertices $\{u,v\}$ with probability $\frac{d_u\cdot d_v}{\sum_{r\in V}d_r}$.
Note that the graph becomes sparser as the exponent $\beta$ increases.
In this multilayer network, we randomly select a vertex subset $V_c$ with some size,
and plant a clique on $V_c$ (in some specified layers).

To evaluate the performance of optimal solution $p\in \Delta(2^V)$ to $(\bm{\alpha},\bm{\beta})$-\dens in the above multilayer network,
we introduce the following measure, which we refer to as the
(expected) F measure:
\begin{align*}
  \text{F measure} \coloneqq \frac{2\cdot \text{precision}\cdot \text{recall}}{\text{precision}+\text{recall}},
\end{align*}
where
$\text{precision}  \coloneqq\mathbb{E}_{S\sim p}\left[\frac{|S\cap V_c|}{|S|}\right]=\sum_{S\subseteq V}p_S\frac{|S\cap V_c|}{|S|}$ and
$\text{recall}     \coloneqq\mathbb{E}_{S\sim p}\left[\frac{|S\cap V_c|}{|V_c|}\right]=\sum_{S\subseteq V}p_S\frac{|S\cap V_c|}{|V_c|}$.
The F measure approaches to 1 if $p$ tends to be close to $V_c$.

We first investigate the case where a clique is planted in all layers.
We generate $k~(=1,2,3,4,5)$ power-law graphs (i.e., layers) with $\beta=2.3$ on $V$ with $|V|=\text{1,000}$.
Then we randomly select a subset $V_c\subseteq V$ consisting of $10$ vertices, and plant a clique on $V_c$ in all layers.
The performance of optimal solutions to $(\bm{\alpha},\bm{\beta})$-\dens is shown in Figure~\ref{subfig:synthetic1-1}.
As can be seen, for any metric, the F measure is reasonably large for $k\geq 2$,
meaning that our algorithm tends to detect $V_c$ using the information of multiple layers.

\begin{figure}[t]
  \begin{minipage}{.5\linewidth}
    \includegraphics[scale=.32]{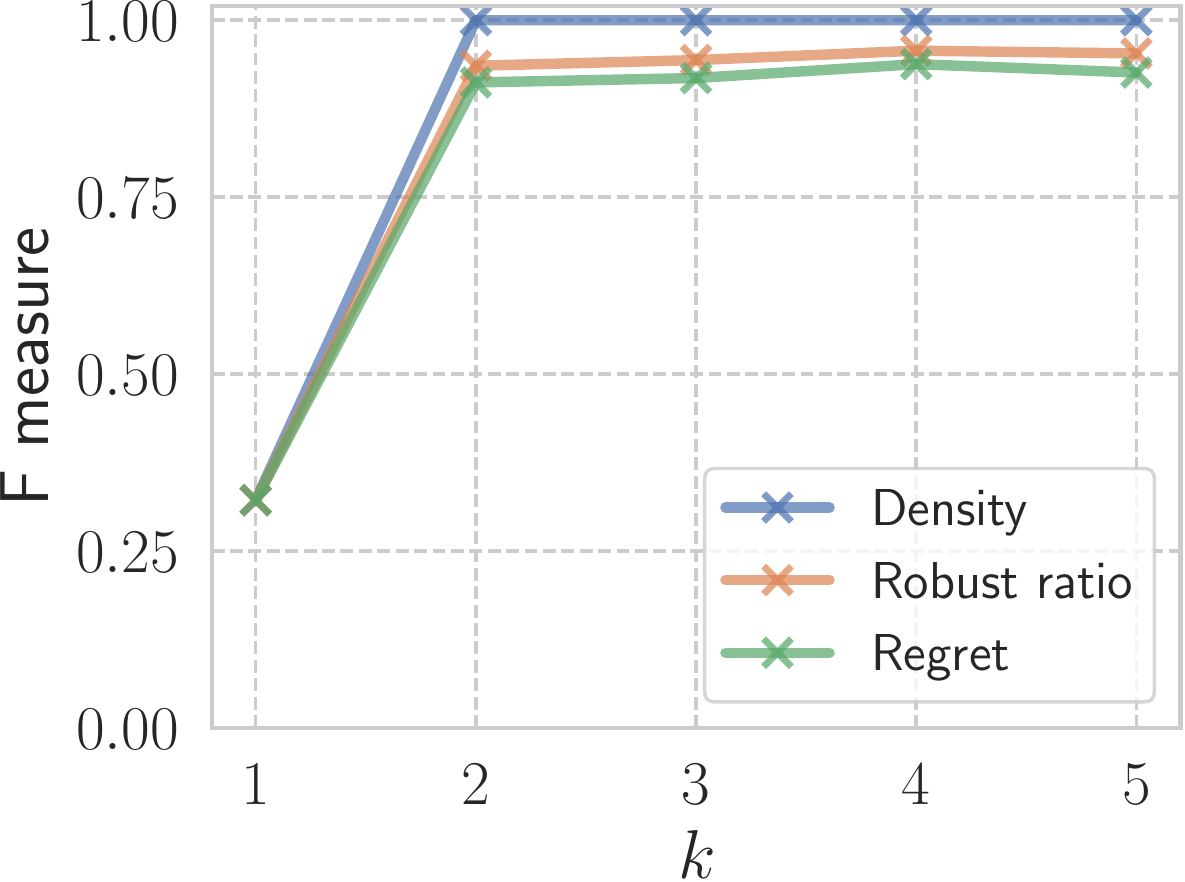}
    \subcaption{Clique in all layers.}\label{subfig:synthetic1-1}
  \end{minipage}%
  \begin{minipage}{.5\linewidth}
    \includegraphics[scale=.32]{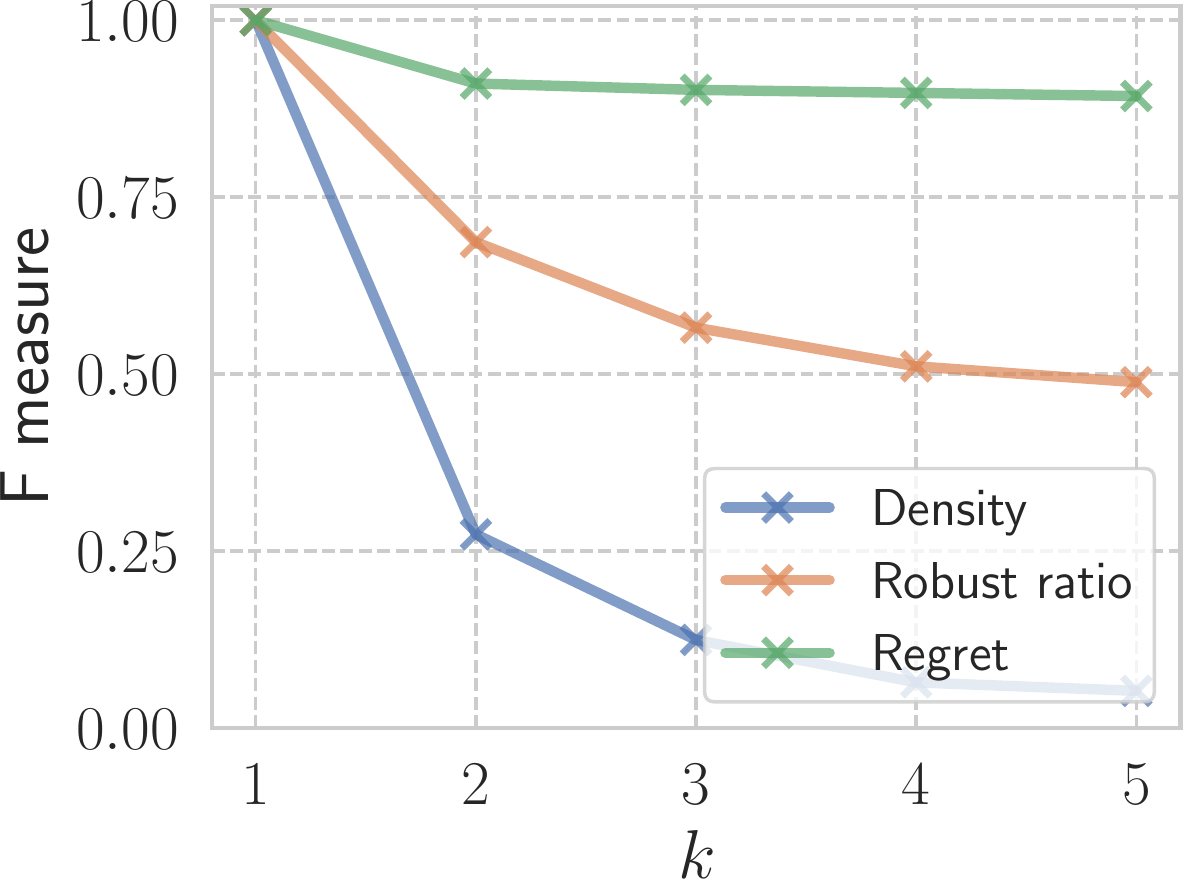}
    \subcaption{Clique in only one layer.}\label{subfig:synthetic1-2}
  \end{minipage}
  \caption{Performance of optimal solutions. Each point corresponds to the average value over 100 network realizations.}\label{fig:synthetic1}
\end{figure}

\begin{figure}[t]
  {
    \begin{minipage}{.5\linewidth}
      \includegraphics[scale=.32]{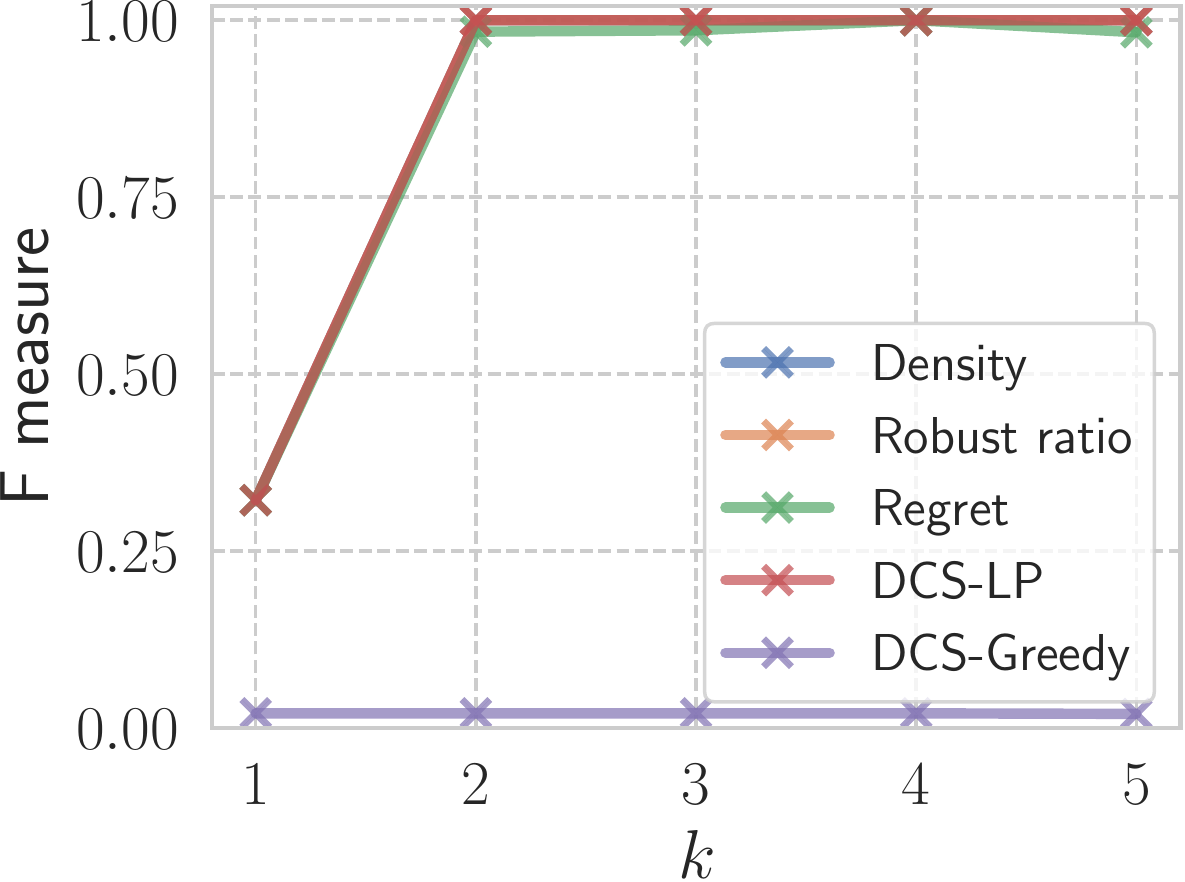}
      \subcaption{Clique in all layers.}\label{subfig:synthetic2-1}
    \end{minipage}%
    \begin{minipage}{.5\linewidth}
      \includegraphics[scale=.32]{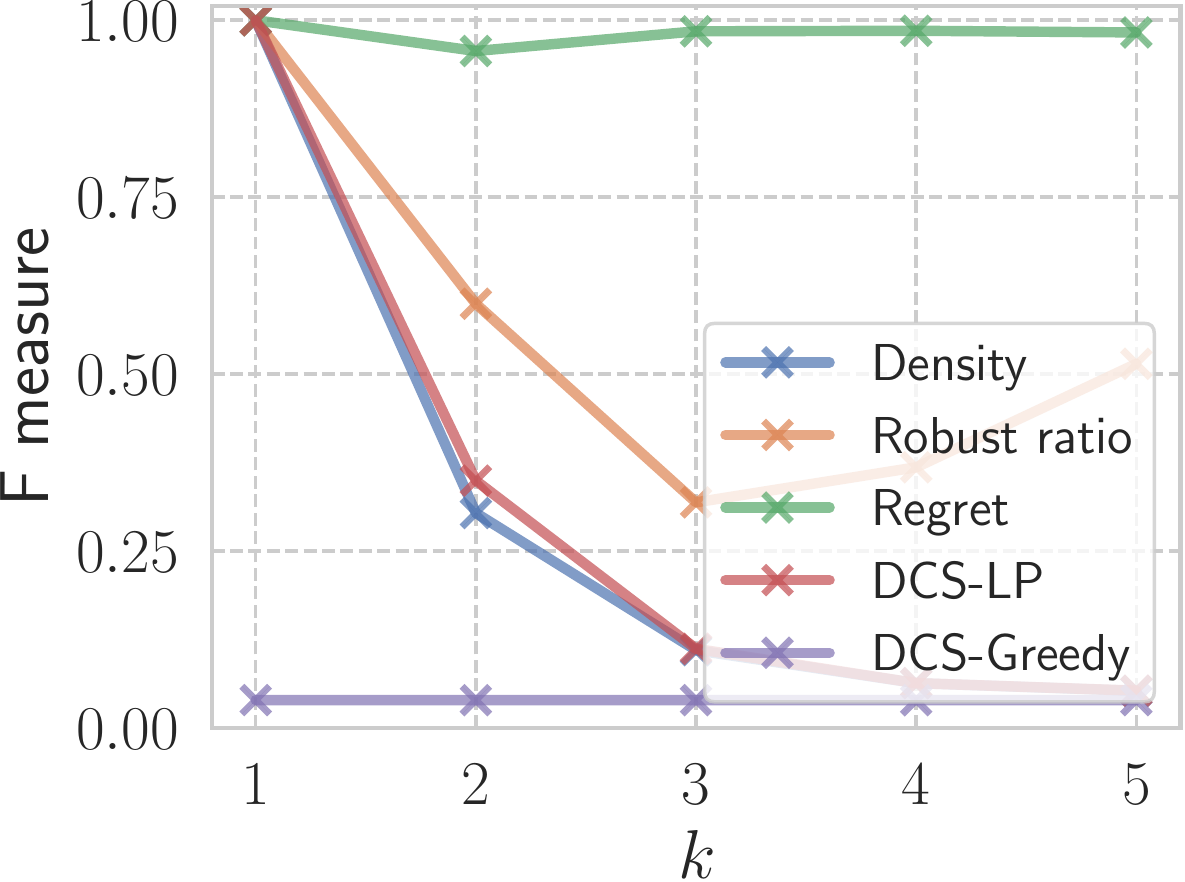}
      \subcaption{Clique in only one layer.}\label{subfig:synthetic2-2}
    \end{minipage}%
  }
  \caption{Performance of single vertex subsets obtained by optimal solutions. The same averaging procedure is applied.}\label{fig:synthetic2}
\end{figure}

\begin{table*}[h!]
  \caption{Performance of our algorithm for real-world datasets.
  OPT and $|\text{supp}|$ denote the optimal value of $(\bm{\alpha},\bm{\beta})$-\dens and the support size of the optimal solution, respectively.
  Note that the objective values for the regret are negated.
  LB indicates the optimal value of LP~\eqref{LP:lower}.
  $|V'|$ and $|E[V']|$, respectively, denote the number of vertices and edges after running Algorithm~\ref{alg:remove}.
  Preprocess time represents the running time of Algorithm~\ref{alg:remove}.
  Total time represents the running time of our proposed algorithm, i.e., Algorithm~\ref{alg:general} together with Algorithm~\ref{alg:remove}.
  For reference, the running time of Algorithm~\ref{alg:general} without Algorithm~\ref{alg:remove} is also presented in the next column.
  The last three columns report the performance of single vertex subsets obtained by optimal solutions.
  For each instance and metric, the best value among the algorithms is in bold.
  }\label{table:computation}
  \setlength{\tabcolsep}{3pt}
  \renewcommand{\arraystretch}{1}
  \scalebox{.78}{
    \begin{threeparttable}
      \begin{tabular}{c|rrr|c|rr|rrrr|r|r|rrr}\toprule
        Dataset                                     & $|V|$   & $|E|$   & $k$ & Metric       & OPT       & $|\text{supp}|$ & LB        & $|V'|$  & $|E[V']|$ & \shortstack{\small Preprocess                                                                                 \\ \small time (s)} & \shortstack{\small Total\\\small time (s)} & \shortstack{\small w/o preprocess\\\small time (s)} & Density  & \shortstack{Robust             \\ratio} & Regret  \\\midrule
                                                    &         &         &     & Density      & 1.1950    & 3               & 1.1313    & 123     & 3,126     & 0.04                          & 0.66     & 0.73     & \textbf{1.1875}  & 0.6579          & -0.6282            \\
        WILDBIRDS                                   & 202     & 4,574   & 6   & Robust ratio & 0.7707    & 4               & 0.7364    & 121     & 2,901     & 0.05                          & 0.77     & 0.79     & 1.0058           & \textbf{0.6970} & \textbf{-0.5065}   \\
        \cite{FS2015}\tnote{$*$}                    &         &         &     & Regret       & -0.4122   & 2               & -0.4473   & 124     & 2,976     & 0.04                          & 0.71     & 0.77     & 0.9129           & 0.6518          & -0.5821            \\\hline
                                                    &         &         &     & Density      & 4.7023    & 2               & 4.4257    & 730     & 6,403     & 1208.79                       & 1304.28  & 1834.88  & \textbf{4.6316}  & 0.6399          & -3.1581            \\
        AS-733                                      & 7,716   & 24,179  & 733 & Robust ratio & 0.7721    & 4               & 0.7295    & 278     & 3,452     & 1165.87                       & 1211.82  & 1906.21  & 4.5789           & 0.7057          & -1.9621            \\
        \cite{LKF2005}\tnote{$\dagger$}             &         &         &     & Regret       & -6.9861   & 4               & -1.8023   & 228     & 3,020     & 1289.15                       & 1329.56  & 1950.33  & 3.6154           & \textbf{0.7317} & \textbf{-1.6316}   \\\hline
                                                    &         &         &     & Density      & 12.0263   & 1               & 12.0263   & 118     & 1,815     & 0.29                          & 0.92     & 8.31     & \textbf{12.0263} & 0.9666          & -0.4382            \\
        Oregon1                                     & 11,492  & 26,461  & 9   & Robust ratio & 0.9808    & 3               & 0.9738    & 110     & 1,702     & 0.48                          & 1.02     & 7.55     & 11.9016          & \textbf{0.9728} & \textbf{-0.3578}   \\
        \cite{LKF2005}\tnote{$\dagger$}             &         &         &     & Regret       & -0.2544   & 3               & -0.3441   & 110     & 1,702     & 0.50                          & 0.95     & 9.72     & 11.8889          & \textbf{0.9728} & \textbf{-0.3578}   \\\hline
                                                    &         &         &     & Density      & 22.2503   & 2               & 21.8228   & 1,130   & 14,166    & 1.90                          & 4.31     & 93.30    & \textbf{18.4000} & 0.1183          & -247.6000          \\
        MoscowAthletics2013                         & 88,804  & 186,846 & 3   & Robust ratio & 0.4709    & 3               & 0.3566    & 105     & 787       & 2.02                          & 2.12     & 112.50   & 13.2000          & \textbf{0.4197} & \textbf{-181.1667} \\
        \cite{ODMA2015}\tnote{$\ddagger$}           &         &         &     & Regret       & -125.4477 & 2               & -130.1752 & 88,804  & 186,846   & 1.40                          & 336.56   & 334.43   & 0.4286           & 0.0186          & -200.8333          \\\hline
                                                    &         &         &     & Density      & 3.5649    & 2               & 3.5077    & 17,551  & 184,659   & 1.98                          & 332.02   & 5466.73  & \textbf{3.5625}  & 0.1768          & -17.0770           \\
        NYClimateMarch2014                          & 102,439 & 329,474 & 3   & Robust ratio & 0.6661    & 3               & 0.5503    & 3,901   & 78,126    & 2.36                          & 147.40   & 6058.12  & 2.3839           & \textbf{0.6652} & -6.8047            \\
        \cite{ODMA2015}\tnote{$\ddagger$}           &         &         &     & Regret       & -2.3052   & 3               & -3.5835   & 102,439 & 329,473   & 1.58                          & 18575.97 & 18909.39 & 1.2785           & 0.3576          & \textbf{-2.3221}   \\\hline
                                                    &         &         &     & Density      & 60.7462   & 2               & 60.7462   & 375     & 4,536     & 11.07                         & 11.72    & 4268.72  & \textbf{52.7143} & 0.0840          & -837.8095          \\
        Cannes2013                                  & 438,537 & 848,017 & 3   & Robust ratio & 0.3633    & 3               & 0.3633    & 246     & 1,617     & 11.53                         & 11.72    & 3981.18  & 42.6000          & \textbf{0.2441} & -365.0667          \\
        \cite{ODMA2015}\tnote{$\ddagger$}           &         &         &     & Regret       & -132.8628 & 2               & -132.8628 & 438,537 & 848,017   & 7.31                          & 4021.19  & 4022.75  & 0.6667           & 0.0073          & \textbf{-155.5000} \\\hline
                                                    &         &         &     & Density      & 1.1891    & 10              & 1.1236    & 191,074 & 559,628   & 18.52                         & 7826.41  & 13787.37 & 0.4146           & 0.0387          & -31.2295           \\
        DBLP                                        & 513,627 & 888,353 & 10  & Robust ratio & 0.1178    & 10              & 0.1125    & 317,231 & 687,335   & 17.01                         & 15323.66 & 22828.17 & \textbf{0.6067}  & \textbf{0.0607} & -20.2759           \\
        \cite{Galimberti+_17}\tnote{$\mathsection$} &         &         &     & Regret       & -11.6944  & 3               & -11.6963  & 513,627 & 888,353   & 14.78                         & 37266.95 & 37085.30 & 0.1045           & 0.0114          & \textbf{-13.4481}  \\
        \bottomrule
      \end{tabular}
      \begin{tablenotes}[para]
        \item[$*$] \url{http://networkrepository.com}
        \item[$\dagger$] \url{http://snap.stanford.edu}
        \item[$\ddagger$] \url{https://comunelab.fbk.eu/data.php}
        \item[$\mathsection$] \url{https://goo.gl/8741Gs}
      \end{tablenotes}
    \end{threeparttable}
  }
\end{table*}

Next we investigate the case where a clique is planted in only one layer.
We generate $k~(=1,2,3,4,5)$ power-law graphs with $\beta=3.0$ on $V$ with $|V|=\text{1,000}$.
Then we randomly select $V_c\subseteq V$ consisting of 20 vertices, but plant a clique on $V_c$ only in one randomly selected layer.
The performance of optimal solutions are described in Figure~\ref{subfig:synthetic1-2}.
As can be seen, the F measure becomes smaller as the number of layers increases.
Among the three metrics, the regret performs particularly well because it concentrates on the layer containing the clique from its definition;
therefore, the regret metric seems most suitable for robust analysis with noisy layers.
The robust ratio performs second best; it also cares about the layer containing the clique.

Finally we conclude this subsection by evaluating vertex subsets that are obtained from optimal solutions $p\in \Delta(2^V)$.
To verify the validity of our model, we select a vertex subset attaining the highest probability in $p\in \Delta(2^V)$.
As our algorithm does not know about $V_c$, it cannot select a vertex subset using the F measure.
For reference, we also run two baseline algorithms, DCS-LP and DCS-Greedy, designed by Jethava and Beerenwinkel~\cite{JB2015}.
Note that DCS-LP can be seen as the algorithm that selects a vertex subset with the largest minimum density value over layers from the support of $p\in \Delta(2^V)$,
where $p\in \Delta(2^V)$ is an optimal solution to our algorithm with the density metric.
The results are depicted in Figure~\ref{fig:synthetic2}, where we employed the same experimental settings as above and used the usual (deterministic) F measure for evaluation.
As can be seen, the trend of our algorithm is similar to that observed in the above experiments;
all metrics almost detect $V_c$ for the all-layers setting, but only the regret metric is successful for the only-one-layer setting.
As for the baseline methods, DCS-LP is comparable to our algorithm with the density metric, while DCS-Greedy performs quite poorly.

\subsection{Performance of our algorithm}
Here we examine the performance of our proposed algorithm using publicly-available real-world multilayer networks.
Table~\ref{table:computation} shows the results together with the summary of the characteristics of the datasets.
Even for large networks,
our algorithm obtains an optimal solution in reasonable time.
The preprocessing algorithm often reduces the size of the networks significantly using a reasonably large lower bound computed by LP~\eqref{LP:lower} 
and results in a substantial speed-up.
As an extreme example, for Cannes2013 with the density and robust ratio metrics, the lower bound attains the optimal value, and the number of vertices is reduced by more than 99.9\%,
which makes the computation more than $300$ times faster.
Consistent with our theoretical analysis, $|\text{supp}|$ is at most $k$. 
For AS-733, $k$ is quite large but $|\text{supp}|$ is still small.

Finally we evaluate the single vertex subsets obtained from optimal solutions.
For $p\in \Delta(2^V)$, we select a vertex subset from $\supp(p)$
that optimizes the metric employed in the algorithm.
Note that in this setting, the output of DCS-LP coincides with that obtained by our algorithm with the density metric.
As DCS-Greedy performed quite poorly, it is omitted.
The results are shown in the last three columns of Table~\ref{table:computation}.
Although there are a few exceptions, the algorithm with a metric performs best
in terms of the metric employed.
A critical fact is that depending only on the density metric, we may fail to obtain meaningful structure from networks.
For example, the algorithm with the robust ratio admits a particularly large robust ratio value of $0.6652$ for NYClimateMarch2014,
meaning that the vertex subset obtained achieves an approximation ratio of $0.6652$ for all layers.
Moreover, the algorithm with the regret metric admits a particularly small regret value of 155.5000 for Cannes2013.
For those instances, the algorithm with the density metric (i.e., DCS-LP) performs poorly in terms of those metrics, respectively.
From the above, it seems quite important to select an appropriate metric depending on the practical purpose at hand.

\section{Conclusion}\label{sec:conclusion}

In this paper, we have introduced a novel optimization model and algorithms for dense subgraph discovery in multilayer networks. 
There are several possible directions for future research.
One direction is to improve the scalability of our algorithm, particularly for the regret metric, 
for which our preprocessing algorithm does not necessarily perform well. 
Another direction is to apply our model and algorithms to some real-world applications of multilayer-network analysis. 
Investigating multilayer-network counterparts of some existing generalizations of the densest subgraph problem (see e.g., \cite{Kawase_Miyauchi_17,Miyauchi_Kakimura_18,Tsourakakis_15,Veldt+21}) is also interesting future work.

\begin{acks}
  This work was partially supported by JST PRESTO Grant Number JPMJPR2122 and
  JSPS KAKENHI Grant Numbers JP17K12646, JP19K20218, JP20K19739, JP21K17708, and JP21H03397.
\end{acks}

\bibliographystyle{ACM-Reference-Format}  
\bibliography{ref}  

\end{document}